\documentclass[aps,pra,twocolumn,floatfix,nofootinbib,superscriptaddress,graphics,longbibliography]{revtex4-2}

\usepackage{bm,graphicx,mathrsfs,amsmath,amssymb,mathtools,makecell,bbm, amsthm,dsfont,color,times,txfonts,nicefrac,framed,float,enumitem,tikz,physics,wrapfig,amsfonts,tcolorbox}
\usepackage{physics}
\usepackage{nicematrix}

\usepackage[colorlinks=true,linkcolor=equationcolor,citecolor=teal]{hyperref}

\hypersetup{urlcolor=teal}

\definecolor{asparagus}{rgb}{0.0, 0.65, 0.31}

\definecolor{main}{RGB}{66,125,58}
\definecolor{sub}{RGB}{221,244,217}

\tcbset{
    sharp corners,
    colback = white,
    before skip = 0.2cm,    
    after skip = 0.5cm      
}   

\newtcolorbox{boxH}{
    sharpish corners, 
    colback = sub, 
    colframe = main, 
    boxrule = 0pt, 
    leftrule = 2pt 
}

\usepackage{thmtools}
\usepackage{thm-restate}

\definecolor{equationcolor}{RGB}{222,94,100}
\definecolor{alecolor}{RGB}{238,33,80}

\hypersetup{urlcolor=equationcolor, colorlinks=true,citecolor=equationcolor,linkcolor=equationcolor} 

\hypersetup{urlcolor=teal}
\definecolor{equationcolor}{RGB}{222,94,100}

\renewcommand{\v}[1]{\ensuremath{\boldsymbol #1}}
\newcommand{\ms}[1]{\textsf{#1}}

\newcommand{\iden}{\mathbbm{1}}

\newcommand{\E}[1]{\mathcal{E}}

\def\E{ {\cal E} }

\def\S{ {\cal S} }

\def\T{ {\cal T} }

\theoremstyle{plain}
\newtheorem{thm}{Theorem}
\newtheorem{lem}[thm]{Lemma}
\newtheorem{prop}[thm]{Proposition}
\newtheorem{conj}[thm]{Conjecture}
\newtheorem{cor}[thm]{Corollary}
\newtheorem{obs}[thm]{Observation}
\newtheorem{defn}{Definition}

\begin{document}

\preprint{APS/123-QED}

\title{Catalytic transformations for thermal operations}

\author{Jakub Czartowski}
	\affiliation{Doctoral School of Exact and Natural Sciences, Jagiellonian University, 30-348 Kraków, Poland}
	\affiliation{Faculty of Physics, Astronomy and Applied Computer Science, Jagiellonian University, 30-348 Kraków, Poland}
 \author{A. de Oliveira Junior}
	\affiliation{Center for Macroscopic Quantum States bigQ, Department of Physics,
Technical University of Denmark, Fysikvej 307, 2800 Kgs. Lyngby, Denmark}

\date{August 22, 2024}

\begin{abstract}
What are the fundamental limits and advantages of using a catalyst to aid thermodynamic transformations between quantum systems? In this work, we answer this question by focusing on transformations between energy-incoherent states under the most general energy-conserving interactions among the system, the catalyst, and a thermal environment. The sole constraint is that the catalyst must return unperturbed and uncorrelated with the other subsystems. More precisely, we first upper bound the set of states to which a given initial state can thermodynamically evolve (the catalysable future) or from which it can evolve (the catalysable past) with the help of a strict catalyst. Secondly, we derive lower bounds on the dimensionality required for the existence of catalysts under thermal process, along with bounds on the catalyst's state preparation. Finally, we quantify the catalytic advantage in terms of the volume of the catalysable future and demonstrate its utility in an exemplary task of generating entanglement and cooling a quantum system using thermal resources.
\end{abstract}

\maketitle


\section{Introduction \label{Sec:introduction}}

Classical thermodynamics is a theory of macroscopic systems in equilibrium, whose description relies on quantities with fluctuations negligible compared to their average values~\cite{fermi1956thermodynamics,callen1985thermodynamics}. As a byproduct, it gives a clear picture of what state transformations are allowed in terms of a small number of macroscopic quantities, such as work and entropy. Specifically, at constant temperature and volume, transformations between equilibrium states are governed by a simple function, the so-called Helmholtz free energy~\cite{Helmholtz1891}. However, going beyond the original scenario of equilibrium thermodynamics has profound consequences. For finite-dimensional quantum systems far from equilibrium, state transformations are no longer governed by a single second law but by an entire family of conditions known as the ``second laws of quantum thermodynamics''~\cite{brandao2015second}. 

Under this paradigm, the existence of a family of generalised free energies, which indicate what states can be reached from a given one (referred to as the present state), naturally decomposes the space of states into three distinct parts: the set of states to which the present state can evolve, known as the future thermal cone $\mathcal{T}_+$; the set of states from which the present state can be reached, or the past thermal cone $\mathcal{T}_-$; and the set of states that are neither in the past nor the future thermal cone, constituting the incomparable region $\mathcal{T}_{\emptyset}$~\cite{Korzekwathesis,de2022geometric} -- three regions analogous to the structure of the light cone in special relativity, which divides the space–time into future, past, and space-like region~(see Fig.~\ref{F-scheme}). Consequently, the issue of thermodynamic transformations in finite-size systems is characterised by the emergence of a finite-size effect known as incomparability. 

Interestingly, one can reduce this effect by introducing a catalyst -- a system that allows us to perform otherwise impossible transformations, but is not modified itself, i.e., at the end of the process it is returned unchanged~\cite{Datta2022,Patrykreview}. Remarkably, if we allow the catalyst to become correlated with the main system while keeping its local state intact, transformations are once again described by a single function, the standard nonequilibrium free energy~\cite{muller2018correlating,wilming2017axiomatic, rethinasamy2020relative, shiraishi2021quantum}. 
\begin{figure}[t]
    \centering
    \vspace{1.1cm}
    \includegraphics{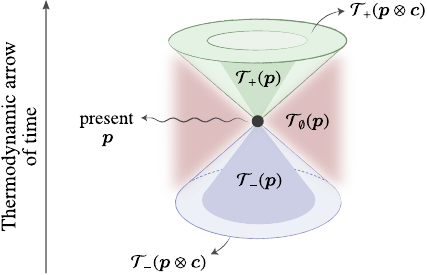}
    \caption{\textbf{Thermodynamic arrow of time}. The second law of thermodynamics introduces an ordering in the space of states such that, for a given initial state, it divides the space into past (blue), incomparable (red), and future (green) regions. By allowing a catalyst $\v c$, the future and past regions grow as the incomparable region shrinks.}
    \label{F-scheme}
\end{figure}

Over the last two decades, 
catalytic transformations
have become increasingly relevant in the context of quantum information science~\cite{Datta2022,Patrykreview}. Their first appearance focused on entanglement manipulation~\cite{Plenio1999,vanDam2003,turgut2007catalytic,daftuar2001mathematical,sun2005existence,feng2005catalyst,PhysRevLett.127.080502,Kondra_2021,Datta2022} which then broadened to include quantum thermodynamics~\cite{brandao2015second,Ng_2015,Wilming2017,muller2018correlating,Lipka-Bartosik2021,shiraishi2021quantum,Gallego2016,boes2019bypassing,Henao_2021,Henao2022,son2022catalysis,son2023hierarchy,Lipka_Bartosik_2023,czartowski2023thermal}, and many other facets of quantum theory~\cite{Aberg2014,Vaccaro2018,lostaglio2019coherence,takagi2022correlation,char2023catalytic,van2023covariant, Marvian2019,wilming2021entropy,Wilming2022correlationsin,rubboli2022fundamental,boes2019neumann,Marvian2019,wilming2021entropy,Wilming2022correlationsin,rubboli2022fundamental,boes2019neumann,de2023quantum}. The only constraint of a catalytic process--returning the catalyst unperturbed--leads to the classification of catalytic processes into two main different types: \emph{(i)} allowing it to become correlated with the main system, or \emph{(ii)} keeping it uncorrelated.
As pointed out before, the former case of correlated catalysis is powerful enough to close the gap between the future and past thermal regions. The latter, referred to as `strict catalysis', is much closer to the original idea of returning the catalyst unperturbed. It is thus natural to ask about limitations of strict catalysts -- a problem we will discuss in this manuscript.

Arguably, one of the fundamental problems in quantum thermodynamics concerns quantification of the catalytic advantages in a given transformation. This issue can be approached by investigating the behaviour of the thermal cones when a catalyst is introduced. In this scenario, the catalyst cannot change future into the past or vice-verse, nor can it shrink both regions. Consequently, the growth of the past and future thermal cones occurs at the expense of the incomparable region. However, determining the set of states that can be achieved with the aid of a catalyst has been considered a highly non-trivial problem. 

A complementary problem is that almost nothing is known about the precise form of the catalyst state required for a given transformation. Although in many cases the necessary (and sometimes sufficient) conditions for the existence of a catalyst are known, the theory remains agnostic regarding the precise form and properties of the catalytic state itself. The first steps towards addressing this problem were made in Ref.~\cite{Sanders_2009, grabowecky2018bounds} within the realm of entanglement theory, where the authors focused on pure state transformations governed by entanglement-assisted local operations and classical communication (ELOCC). Therein, the authors consider bounds for the entanglement of a catalyst necessary to catalyse transformation between a given pair of otherwise incomparable states. Furthermore, they established a lower bound for the dimension of the catalyst required for a specific ELOCC transformation. As these results rely on the notion of majorisation, they automatically extend to the resource theory of coherence, as well as to so-called noisy operations, or incoherent thermal operations at infinite temperature. For finite temperatures, state transformations between specific quantum states (those that are diagonal in the energy eigenbasis) are described by a more general notion known as thermomajorisation order. This concept does not have an equivalent physical counterpart in either entanglement or coherence theories. Currently, understanding any relationship between the dimensionality or form of the catalyst in a given thermal process aided by a strict catalyst remains an open problem.

In this work, we extend our understanding of the interplay between quantum thermodynamics and catalysis. Our main question can be framed as:
\begin{boxH}
		\emph{What are the ultimate limits for the set of states achievable from a given initial state, with the aid of catalyst, under the most general energy-conserving interactions between the system and the thermal bath?}	
\end{boxH}

This question can be tackled by defining the \emph{catalytic future thermal cone $\mathcal{T}_{\mathcal{C}+}$}. This set represents the states to which a given initial state $\v{p}$ can be transformed with the aid of a catalyst. Similarly, one can also define the \emph{catalytic past thermal cone} $\mathcal{T}_{\mathcal{C}-}$ by considering $\v{p}$ as the target state of a transformation to be catalysed. We shed light on the answer to this question by bounding the catalytic past and future cones -- in other words, we construct a \emph{catalysable past} and \emph{future}. More precisely, 
for a $d$-dimensional energy-incoherent state, we present an explicit construction of the catalysable past $\mathcal{C}_{-}$ and the catalysable future $\mathcal{C}_{+}$, along with a characterization of its extreme points. These results reveal fundamental limits on the advantages achievable with uncorrelated catalysts~[see Fig.~\ref{Fig:warming-up-example} for a warm-up example, where we present the future and past regions with and without catalyst for a three-and four-level systems]. As these results are general and do not rely on any specific assumptions regarding the catalytic state, our second focus is on the dimensionality and populations of the catalyst. We derive lower bounds on the dimensionality required for catalysts under thermal operations, together with bounds on the catalysts population. Our results allows us to construct a catalytic state given two incomparable states. Surprisingly, our findings show that the no-go result from~\cite{grabowecky2018bounds} formulated for the resource theory of entanglement and applicable for the infinite-temperature setting, stating that catalysis has no effect under local operations and classical communication for a main system of dimension three, ceases to apply in the context of thermodynamics at finite temperatures.
 \begin{figure}[t]
    \centering
    \includegraphics{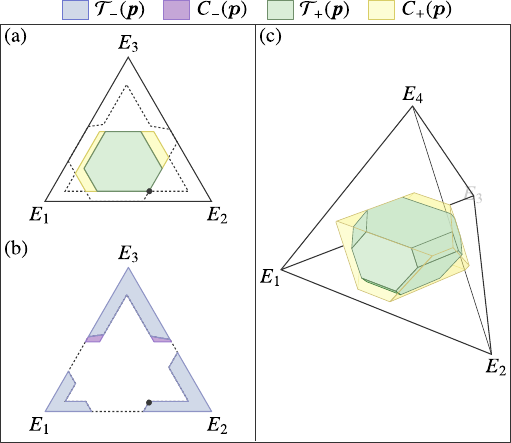}
    \caption{\textbf{Catalysable regions.} 
    For a three-level system with population given by $\v p = (0.42,0.51,0.07)$, represented by a black dot~$\bullet$, and energy spectrum $E_1 = 0, E_2 = 1$, and $E_3 = 2$ for $\beta = 0.2$, we plot its (a) catalysable future $C_+(\v p)$ (yellow) and (b) past  $C_{-}(\v p)$ (purple), together with their respective past and future thermal cones. In (c) we show the catalysable future for the four-level system with $\v p = (0.05,0.28,0.40,0.27)$ and energy spectrum $E_1 = 0, E_2 = 1$, $E_3 = 2$, and $E_4 = 3$.}
    \label{Fig:warming-up-example}
\end{figure}

Using the above framework, we also quantify the catalytic advantage in terms of the volume of the catalysable future. In this context, we offer a detailed discussion on its behaviour as a function of the ambient temperature. As an application, we demonstrate the benefits of employing a catalyst in the generation of entanglement. Recently Ref.~\cite{de2024ent} identified the set of states that cannot become entangled under thermal operations. In this analysis, we elucidate the catalytic advantages by identifying the set of states that, once catalysed, enable entanglement generation under thermal operations.

The paper is organised as follows. In Section~\ref{Sec:framework}, we introduce the resource-theoretic approach to thermodynamics, review known results concerning the conditions for state transformation under thermal operations, and establish the notation of the paper. Section~\ref{Sec:results}, collecting our main results, begins
with
the construction of the catalysable past and future regions, giving a simple method for generating all the extreme points of the catalysable future, thus providing the closest analogue to the Birkhoff theorem for the problem of catalysis in thermal operations. Next, we state our second main result, which establishes a lower bound on the dimensionality of the catalyst and introduces a constraint on its population vector for any pair of incomparable states. In Section~\ref{sec:volume} we discuss a quantification of the catalytic advantage by volume of the catalysable future and past regions. Then, in Section~\ref{Sec:applications-1}, we provide an application of our results to entanglement generation under thermal operations. Then, in Section~\ref{sec:cool} we provide additional application in catalytic cooling. Finally, we conclude with an outlook in Section~\ref{sec:outlook}. The technical derivation of our results can be found in the Appendices.


\section{Framework \label{Sec:framework}}

In this work, our aim is to investigate how the structure of thermal cones is modified when a catalyst is introduced. We approach this problem using the resource-theoretic framework of thermodynamics~\cite{gour2015resource,Lostaglio2019}. In this section, we briefly review well-known results without going into details and set the scene by introducing the relevant notation.

\subsection{Thermal operations}
We consider a composite system comprised of a finite-dimensional system and a thermal environment at temperature $T = \beta^{-1}$, where we set $k_B = 1$. The main system is described by a Hamiltonian $H = \sum^{d}_{i=1} E_i\ketbra{E_i}{E_i}$ and is prepared in a state $\rho$; while, the thermal environment, with Hamiltonian $H_{\ms R}$, is assumed to be in a thermal equilibrium state, $\gamma_{\ms R} = e^{-\beta H_{\ms R}}/Z_{\ms R}$, where $Z_{\ms R} = \tr(e^{-\beta H_{\ms R}})$ is the partition function. 

The evolution of the composite system is modeled by considering the set of \emph{thermal operations} (TOs)~\cite{Janzing2000,horodecki2013fundamental,Lostaglio2019}. These are quantum channels defined upon minimal assumptions, such as that the joint system is closed and evolves via an energy-preserving unitary. This is captured by completely positive trace-preserving (CPTP) maps that act on $\rho$ as
	\begin{equation}
		\label{Eq:thermal_operations}
		\E(\rho)=\Tr_{\ms R}\left[U\left(\rho \otimes \gamma_{\ms R}\right)U^{\dagger}\right],
	\end{equation}
	where $U$ is a joint unitary that commutes with the total Hamiltonian of the system and the bath $[U, H\otimes \iden_{\ms R}+ \iden\otimes H_{\ms R}] = 0$.
 
The fundamental question within the resource-theoretic approach is to identify the set of states that a given state $\rho$ can be transformed to under thermal operations. The reachability of states under TOs can
be studied by introducing the notion of thermal cones~\cite{de2022geometric}. Given $\rho$, the set of states achievable under a thermal operation is referred to as the future thermal cone $\T_+(\rho)$. Conversely, the set of states that can evolve into $\rho$ is called the past thermal cone $\T_-(\rho)$. The set of states that are neither in the past nor the future of $\rho$ is termed the incomparable thermal region $\T_{\emptyset}(\rho)$.  The general characterization of thermal cones is not known beyond the simplest qubit case~\cite{lostaglio2015quantum,PRLHorodeckicoherence}.

However, for states that are block-diagonal in the energy eigenbasis, also known as \emph{energy-incoherent states}, a simple construction based on thermomajorisation relation exists. For this class of states, the existence of a thermal operation between two states $\rho$ and $\sigma$ is equivalent to the existence of Gibbs preserving (GP) matrices acting on their eigenvalues $\v p = \operatorname{eig(\rho)}$ and $\v q = \operatorname{eig}(\sigma)$~\cite{Janzing2000,horodecki2013quantumness}. Furthermore, the existence of a GP matrix connecting two states can be expressed by thermomajorisation relations between $\v p$ and $\v q$~\cite{horodecki2013fundamental}. 

The framework of thermal operations naturally incorporates the phenomenon of catalysis by simply assuming that the system $\rho$ is given by a joint and uncorrelated system $\rho = \rho_{\ms S}\otimes \omega_{\ms C}$ with Hamiltonian $H = H_{\ms S}\otimes \iden_{\ms C}+\iden_{\ms S}\otimes H_{\ms C}$\footnote{It is known in the literature that, without loss of generality, one can
always choose a trivial Hamiltonian for the catalyst, $H_\ms{C}\propto \iden$~\cite{horodecki2013fundamental}. However, since our results will not rely on this and, furthermore, in practical settings one is unlikely to have a trivial Hamiltonian, we will not use this assumption.}. Thus, we consider catalytic thermal operations (CTOs) to be transformations of the following form
\begin{equation}\label{Eq:CTO}
    \mathcal{E}(\rho \otimes \omega_{\ms C}) = \sigma\otimes \omega_{\ms C}.
\end{equation}
When $\rho$ and $\sigma$ are energy-incoherent states, the necessary conditions for the existence of a transformation as in Eq.~\eqref{Eq:CTO} is captured by a set of quantities called $\alpha$-free energies
\begin{equation}
    F_{\alpha}(\rho,\gamma) \geq F_{\alpha}(\sigma,\gamma) \:\:, \:\: \forall \alpha \geq 0,
\end{equation}
where $F_{\alpha}(\rho):=[D_{\alpha}(\rho\|\sigma)-\log Z]/\beta$, with $D_{\alpha}(\rho\|\sigma)$ being the quantum Rényi divergence~\cite{renyi1961measures,6832827}. Nevertheless, a complete characterization of the catalytic future thermal cone has not yet been addressed. 

\subsection{Mathematical preliminaries \label{Sub:Mathematical-preliminaries}}
In our analysis we focus on $d$-dimensional states that are diagonal in the energy eigenbasis and are equivalently described by probability vectors corresponding to populations assigned to respective energy levels,
\begin{equation}
    \rho = \sum^{d}_{i=1} p_i\ketbra{E_i}{E_i} \longmapsto \v p = (p_1, ..., p_d).
\end{equation}
Thus, the states under consideration live in the $d$-dimensional probability simplex,
\begin{equation}\label{Eq:probability-simplex}
    \Delta_d= \qty{\v{p} = (p_1, ..., p_d) \in \mathbbm{R}_{\geq0}^d : \sum_i p_i = 1} .  
\end{equation}
Unless stated otherwise, throughout this manuscript, we will work under the assumption that the energy levels are non-degenerate, i.e., $i\neq j \Rightarrow E_i \neq E_j$.

Furthermore, we define the thermal distribution $\v{\gamma}$ and slope vector $\v{s}(\v{p})$ associated with a probability vector $\v{p}$ as
\begin{align}
    \v \gamma := \frac{1}{Z}(e^{-\beta E_1}, ..., e^{-\beta E_d}) \:\:\: \text{and} \:\:\:
    \v s (\v p) := \underbrace{\qty(\frac{p_1}{\gamma_1}, ..., \frac{p_d}{\gamma_d})^\downarrow}_{\qty(s_1(\v{p}),\hdots,s_d(\v{p}))},
\end{align}
where the down arrow denotes the vector arranged in non-increasing order. Next, we define the $\beta$-order of $\v p$: 

\begin{defn}[$\beta$-ordering] Given $\v p$ and $\v \gamma$, the $\beta$-ordering of $\v p$ is defined by a permutation $\v \pi_{\v p}$ that satisfies
\begin{equation} \label{eq_beta-ordering}
    s_{\v \pi_{\v p}(i)}(\v{p}) = \frac{p_i}{\gamma_i}.
\end{equation}
Thus, the $\beta$-ordered version of $\v p$ is given by
	\begin{equation}
	\v{p}^{\, \beta}=\left(p_{\v \pi_{\v{p}}^{-1}(1)},\dots ,p_{\v \pi_{\v{p}}^{-1}(d)}\right).
	\end{equation}
Note that each permutation belonging to the symmetric group, $\v \pi \in \mathcal S_d$, defines a different $\beta$-ordering on the energy levels of the Hamiltonian $H$    
\end{defn}

The above definition allows us to define a thermomajorisation curve:

\begin{defn}[Thermomajorisation curve]
Given $\v p$, a thermomajorisation curve $f^{\beta}_{\v p}: \{0,1\} \to \{0,1\}$ is a piecewise linear function composed of segments connecting the point $(0,0)$ and the points defined by consecutive subsums of the $\beta$-ordered form of the probability $\v{p}^\beta$ and the Gibbs state $\v{\gamma}^\beta$,
		\begin{equation}
		\left(\sum_{i=1}^k\gamma^{\, \beta}_i,~\sum_{i=1}^k p^{\, \beta}_i\right):=\left(\sum_{i=1}^k\gamma_{\v \pi^{-1}_{\v{p}}(i)},~\sum_{i=1}^k p_{\v \pi^{-1}_{\v{p}}(i)}\right),
		\end{equation}
		for $k\in\{1,\dots,d\}$.   
\end{defn}

Finally, we define the notion of thermomajorisation in terms of the respective thermomajorisation curves:
\begin{defn}[Thermomajorisation]
    Given two $d$-dimensional probability distributions $\v p$ and $\v q$, and a fixed inverse temperature $\beta$, we say that $\v p$ \emph{thermomajorises} $\v q$ and denote it as $\v p \succ_{\beta} \v q$, if the thermomajorisation curve $f^{\, \beta}_{\v{p}}$ is above $f^{\, \beta}_{\v{q}}$ everywhere, i.e.,
		\begin{equation}
			\v p \succ_{\beta} \v q \iff \forall x\in[0,1]:~ f^{\, \beta}_{\v{p}}(x) \geq f^{\, \beta}_{\v{q}}(x) \, .
		\end{equation}
\end{defn}
Let us make a few comments about the above definition. First, at the infinite-temperature limit (or when $\beta = 0$), the thermal distribution vector $\v{\gamma}$ becomes the uniform state
\begin{equation}
    \v{\eta} =\frac{1}{d}(1, ..., 1).
\end{equation}
Consequently, the concept of $\beta$-ordering simplifies to arranging the probability vectors in non-increasing order. Thus, thermomajorisation reduces to the well-known concept of majorisation~\cite{marshall1979inequalities}. Second, thermomajorisation does not introduce a total order (or even fails to be a partial order in some cases, as proven in Ref.~\cite{vomEnde2024}), a given pair of states $\v p$ and $\v q$ is said to be incomparable when neither $\v p$ thermomajorises $\v q$, nor $\v q$ thermomajorises $\v p$. We denote incomparable states as $\v p \perp^{\beta} \v q$. 

The concept of incomparability can be studied by introducing the family of tangent vectors, a concept stemming from the idea of tangent function. To gain some intuition about their properties and interpretation, we first present its construction for the simple case of $\beta=0$ and then generalise to $\beta > 0$.

For any vector $\v p$, a tangent vector $\v{t}(\v{p}) \equiv \v{t}$ is defined by imposing that all its components, except the first and the last are equal, $t_i = t_j$ for all $1 < i < j < d$. Furthermore, we require that the majorisation function $f_{\v t}$ agrees with $f_{\v p}$ at no more than two consecutive elbows, i.e., $f_{\v t}(i/d) = f_{\v p}(i/d)$ and $f_{\v t}(x) > f_{\v p}(x)$ for $x\in\qty(0,i/d)\cup(i/d,1)$ or $f_{\v t}(x) = f_{\v p}(x)$ for $x\in[i/d,(i+1)/d]$ and $f_{\v t}(x) > f_{\v p}(x)$ elsewhere. 
The two imposed conditions follow the intuition of tangency and, by construction, satisfy the majorisation relation $\v t \succ \v p$.

\begin{figure}[t]
    \centering

    \includegraphics{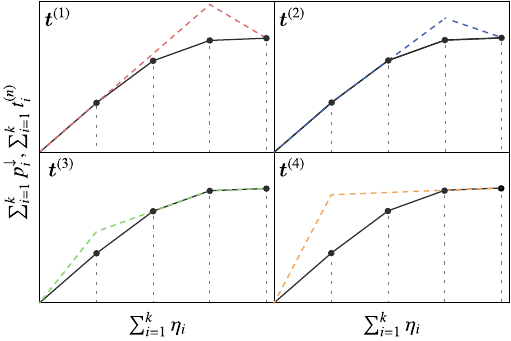}
    \caption{\label{fig-majorisationd4-examples2} \textbf{Majorisation curves of the tangent vectors}. For the initial state $\v p =(0.43, 0.37, 0.18, 0.02)$, we plot the majorisation curve of $\v p$ (black curve) along with the tangent vectors $\v{t}^{(1)}$, $\v{t}^{(2)}$, $\v{t}^{(3)}$ and $\v{t}^{(4)}$.} 
\end{figure}

Assuming equality between $f_{\v t}(x)$ and $f_{\v p}(x)$ on a single linear segment, $x\in[i/d,(i+1)/d]$, 
limits the tangent vectors to a set of $d$ unique probability vectors $\v{t}^{(n)}$, defined as follows:
\begin{equation}
\label{Eq:tangent-vectors}
    \v{t}^{(n)}=\left(t^{(n)}_1, p_n^{\downarrow}, ..., p_n^{\downarrow}, t^{(n)}_d\right) , 
\end{equation}
for $1 \leq n \leq d$, where the first and last components are given by
\begin{align}
\label{eq:extremalpointspastan}
    t_1^{(n)} & = \sum_{i=1}^{n-1} p_i^{\downarrow} - (n-2) p_n^{\downarrow}, &
    t_d^{(n)} & = 1 - t^{(n)}_1 - (d-2)p_n^{\downarrow}.
\end{align}
An exemplary set of tangent vectors is depicted in Fig.~\ref{fig-majorisationd4-examples2}.

Note that the tangent vectors $\v{t}^{(n)}$, which agrees with the majorisation curve of $\v p$ at two successive points, can be used to construct all possible tangent vectors $\v t$ that meet the condition of agreement at least at a single point, $f_{\v t}(i/d) = f_{\v p}(i/d)$.
The fact that $\v{t}^{(n)}$ may be a quasi-probability distribution does not pose a problem as this vector can always be projected back onto the probability simplex.

The intuition behind tangent vectors can now finally be generalised to the case of finite temperatures. This is achieved as follows:

\begin{defn}[Thermal tangent vectors]
Given an energy-incoherent state $\v p$ and a thermal state $\v \gamma$, consider distributions $\v t^{(n,\v \pi)}(\v p) \equiv t^{(n,\v \pi)}$ in their $\beta$-ordered form, constructed for each permutation $\v \pi \in \S_{d}$ and $n \in \{1, ..., d\}$,
\begin{equation}
\label{eq_thermaltangentvectors}
\left[\v{t}^{(n,\v \pi)}\right]^{\, \beta} = \left(
    t^{(n,\v \pi)}_{\v \pi(1)}, s_n
    \gamma_{\v \pi (2)}, ..., 
    s_n \gamma_{\v \pi (d-1)},
    t^{(n,\v \pi)}_{\v \pi(d)}\right) ,  
\end{equation}
with
\begin{subequations}
\begin{align}\label{Eq:tangent-vectors-1}
    t^{(n,\v \pi)}_{\v \pi(1)} & = \sum_{i=1}^n p^{\, \beta}_i-s_n\left(\sum_{i=1}^n \gamma^{\,\beta}_i-\gamma_{\v \pi(1)} \right), \\ 
    t^{(n,\v \pi)}_{\v \pi(d)} & = 1- t^{(n,\v \pi)}_{\v \pi(1)}-s_n\sum_{i=2}^{d-1} \gamma_{\v \pi(i)},
    \label{Eq:tangent-vectors-2}
\end{align}
\end{subequations} 
where $s_n := \v s_n(\v p).$
\end{defn}
Given that the previous definition naturally includes the case of $\beta = 0$, we will, from now on, simply refer to thermal tangent vectors as tangent vectors.

Finally, notice that the same intuition as the case of $\beta =0$ holds for $\beta \geq 0$: a tangent vector $\v{t}^{(n,\pi)}$ is a probability distribution with $\beta$-order given by of almost-constant slope, such that $s[\v{t}^{(n,\pi)}]_i = s_n(\v{p})$ for $1<i<d$, with the first and last slope adjusted such that $\v{t}^{(n,\pi)} \succ_{\beta} \v{p}$ and there exists a point $x$ at which $f_{\v{t}^{(n,\pi)}}(x) = f_{\v{p}}(x)$. For a very large dimension $d$ such a state indeed looks like a tangent at point $x$.


\section{Main results \label{Sec:results}}

In this section, we present our main results, which are divided into two parts. First, we provide a condition that allows one to characterise the set of states that can be catalysable for a given initial state. These results are general and do not rely on any specific assumptions regarding the catalyst's structure, including its dimension, Hamiltonian, or state. Subsequently, the second part of our results complements the first by establishing bounds on dimensionality and the region of the state space within which a catalyst can be found.

\subsection{Catalysable past and future regions}

We start by defining the catalytic thermal cones:

\begin{defn}[Catalytic thermal cones]
        For a $d$-dimensional energy-incoherent state $\v p$, we define the catalytic future/past thermal cone $\mathcal{T}_{\mathcal{C}{+/-}}$ as subsets of the incomparable region $\mathcal{T}_\emptyset(\v{p})$ such that
        \begin{align}
          \mathcal{T}_{\mathcal{C}+}(\v{p})& \equiv \qty{\v{q}\in\Delta_d: \v{q}\perp_\beta \v{p},\exists_{k\geq2,\v{c}\in\Delta_k}: \v{p}\otimes\v{c}\succ_{\beta}\v{q}\otimes\v{c}}, \\
            \mathcal{T}_{\mathcal{C}-}(\v{p})& \equiv \qty{\v{q}\in\Delta_d: \v{q}\perp_\beta \v{p},\exists_{k\geq2,\v{c}\in\Delta_k}: \v{q}\otimes\v{c}\succ_{\beta}\v{p}\otimes\v{c}},
        \end{align}
\end{defn}
\noindent In what follows, if $\v{q}\in\mathcal{T}_\emptyset(\v{p})$, we will say that transformation $\v{p}\rightarrow\v{q}$ is catalysed, or that $\v{q}$ is catalysed, omitting explicit relation to $\v{p}$ whenever it is clear from the context.

The following result establishes a necessary condition for any state $\v{q}$ to belong to the catalytic future thermal cone of a state~$\v{p}$:
\begin{lem}[Catalytic condition] \label{lem:catalysable_future_cond}
    Consider a pair of incomparable states $\v{q}\perp_\beta\v{p}$. The state $\v{q}$ is catalysed (with respect to $\v{p}$), ie. $\v{q}\in\mathcal{T}_{\mathcal{C}+}(\v{p})$ [or, equivalently, $\v{p}\in\mathcal{T}_{\mathcal{C}-}(\v{q})$], only if it satisfies:
    \begin{equation}\label{eq:nec_cond_beta}
         s_1\qty(\v{p}) > s_1\qty(\v{q}) \quad\text{and}\quad s_d\qty(\v{p}) < s_d\qty(\v{q}) .
    \end{equation}
\end{lem}

\begin{figure*}
    \centering
    \includegraphics{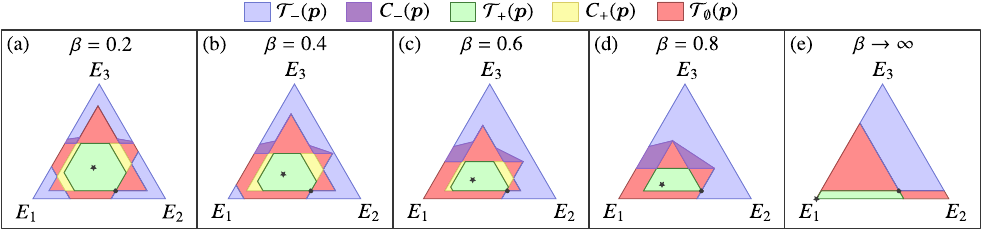}
    \caption{\textbf{Catalysable regions for $d=3$}. For a three-level system with a state given by $\v{p} = (0.34, 0.59, 0.07)$, represented by a black dot~$\bullet$, and a thermal state $\v{\gamma}$, represented by a black star $\bigstar$, with an energy spectrum $E_1 = 0, E_2 = 1$, and $E_3 = 2$, we plot its catalysable thermal cone for $\beta$ ranging from (a) $\beta = 0.2$ to $\beta \to \infty$.}
    \label{Fig:catalysable-thermal-cones-example}
\end{figure*}

\begin{proof}
    Let us take a potential catalyst state $\v{c}\in\Delta_k$. The slope vectors will be given as
    \begin{equation}
        \begin{aligned}
            \v{s}\qty(\v{p}\otimes\v{c}) = \biggl(& 
            s_1(\v{p}) s_1(\v{c}),\hdots,
            s_1(\v{p}) s_k(\v{c}),\hdots, \\
            & s_d(\v{p}) s_1(\v{c}),\hdots,
            s_d(\v{p}) s_k(\v{c})
            \biggr)^{\downarrow},
        \end{aligned}
    \end{equation}
    and same for $\v{q}\otimes \v{c}$. In order for one state to thermomajorise the other, $\v{p}\otimes\v{c}\succ_\beta\v{q}\otimes\v{c}$, we need
    \begin{align}
        s_1\qty(\v{p}\otimes\v{c}) > s_1\qty(\v{q}\otimes\v{c}), &&
        s_{dk}\qty(\v{p}\otimes\v{c}) < s_{dk}\qty(\v{q}\otimes\v{c}).
    \end{align}
    From the definition of $\beta$-ordering we know instantly that $s_i(\v{p}) \geq s_j(\v{p})$ whenever $i \leq j$, similarly for $\v{q}$ and $\v{c}$. Consequently, we can state that
    \begin{align}
        s_1\qty(\v{p}\otimes\v{c}) = s_1(\v{p}) s_1(\v{c}), &&
        s_{dk}\qty(\v{p}\otimes\v{c}) = s_d(\v{p}) s_k(\v{c}),
    \end{align}
    and, again, analogously for $\v{q}$. Thus, we have
    \begin{equation}
        \begin{aligned}
           s_1\qty(\v{p}\otimes\v{c}) > s_1\qty(\v{q}\otimes\v{c}) \Rightarrow
            s_1\qty(\v{p}) > s_1\qty(\v{q}), \\
           s_{dk}\qty(\v{p}\otimes\v{c}) < s_{dk}\qty(\v{q}\otimes\v{c}) \Rightarrow
            s_d\qty(\v{p}) < s_d\qty(\v{q}),
        \end{aligned}
    \end{equation}
    which completes the proof.
\end{proof}
Combining Lemma~\ref{lem:catalysable_future_cond} with the tangent vectors introduced in Sec.~\ref{Sub:Mathematical-preliminaries} enables us to provide a simple upper bound for the catalytic thermal cones and characterise the catalysable regions as:

\begin{thm}[Catalysable regions\label{Thm:catalysable-set}] Consider an energy-incoherent state $\v p$ and its associated tangent vectors $\v{t}^{(n,\pi)}(\v{p})$. We define two auxiliary sets
    \begin{equation}
        T_i(\v p) = \operatorname{conv}\qty[\qty{\v{t}^{(i,\pi)}(\v p)}_{\pi\in\mathcal{S}_d}],
    \end{equation}
    where $\mathcal{S}_d$ is the set of all possible $d$-element permutations. The catalysable future and past regions are bounded by
    \begin{align}
        \mathcal{C}_+(\v p) &= \qty[T_1(\v p)\cap T_d(\v p)]\setminus\mathcal{T}_+(\v p) , \\
       \mathcal{C}_-(\v p) &= \mathcal{T}_\emptyset(\v p)\setminus\qty[T_1(\v p)\cup T_d(\v p)].
    \end{align}
That is, $\mathcal{C}_{-/+}(\v p) \subset \mathcal{T}_{\mathcal{C}-/+}$.
\end{thm}
\begin{proof} 
    The above follows from the fact that $\v{t}^{(i,\pi)}(\v p)$ by construction thermomajorises all $\v{q}$ with $\beta$-order $\v \pi$ such that $s_i(\v{q}) = s_i(\v{p})$. In particular, $\v{t}^{(1,\pi)}(\v p)$ and $\v{t}^{(d,\pi)}(\v p)$ thermomajorise all $\v{q}$ that satisfy first and second part of Lemma~\ref{lem:catalysable_future_cond}, respectively. As a consequence, $\v{q}$ thermomajorised by both $\v{t}^{(d,\pi)}(\v p)$ and $\v{t}^{(d,\pi)}(\v p)$ satisfies Lemma~\ref{lem:catalysable_future_cond}. The proof is completed by noticing, that 
    \begin{equation}\label{Eq:t-vectors}
        \bigcup_{\pi\in\mathcal{S}_d} \mathcal{T}_+[\v{t}^{(i,\pi)}(\v p)] = \operatorname{conv}\qty[\qty{\v{t}^{(i,\pi)}(\v p)}_{\pi\in\mathcal{S}_d}].
    \end{equation}
\end{proof}
Theorem~\ref{Thm:catalysable-set} establishes a fundamental limit for when an incoherent strict catalyst is allowed to be used in a thermodynamic process. In other words, there is no thermodynamic process aided by a strict catalyst that can bring the initial state out of the catalysable future $\mathcal{C}_{+}$. Hence, the advantages gained from using a catalyst are captured by such a region. However, it is important to emphasize the limitations of such a result, namely given a state $\v{q}$ belonging to such a region, whether there exists a catalyst $\v{c}$ and a thermal operation $\mathcal{E}$, such that $\mathcal{E}(\v{p} \otimes \v{c}) = \v{q}$ still remains an open problem.

In what follows, when we appeal to the sets $T_i(\v{p})$, we will drop the argument whenever it is clear from the context.

Finally, the extremal points of the catalysable set can be derived directly from its thermomajorisation curve:
\begin{cor}[Extreme points of $\mathcal{C}_+$\label{Cor:extreme-points-catalysable}] 
    Consider an energy-incoherent state $\v{p}$ and an extreme point $\v{v}^{\pi}$ of its catalysable future region $\mathcal{C}_+(\v{p})$, corresponding to the $\beta$-order $\v{\pi}$. The elbows of its thermomajorisation curve $f_{\v{v}^{\v{\pi}}}$ are determined by
    \begin{equation}
        f_{\v{v}^\pi}\qty(\Gamma_i) = \min\qty[
        f_{\v{t}^{(1,\v{\pi})}(\v p)}\qty(\Gamma_i),f_{\v{t}^{(d,\v{\pi})}(\v p)}\qty(\Gamma_i)],
    \end{equation}
    where we have $\Gamma_i = \sum_{j=1}^i \gamma_{\pi(j)}$.
\end{cor}
\begin{proof}
    It follows directly from translating Theorem~\ref{Thm:catalysable-set} to thermomajorisation curves.
\end{proof}
It is interesting to note that the full catalysable future $\mathcal{C}_+$ of any given state will have much fewer extreme points than the corresponding future thermal cone.

\begin{lem}[Number of extreme points]
    Consider a state $\v{p}\in\Delta_d$. The number of extreme points of the full catalysable future thermal cone $\mathcal{CT}_+(\v{p}) \equiv \mathcal{T}_+(\v{p})\cup \mathcal{C}_+(\v{p})$ of any given state $\v{p}\in\Delta_d$ is upper-bounded by
    \begin{equation}
        \abs{\operatorname{Ext}(\mathcal{CT}_+(\v{p}))} \leq \left\lceil\frac{d}{2}\right\rceil\binom{d}{\left\lceil\frac{d}{2}\right\rceil} \ll d!,
    \end{equation}
    where $d! \geq \abs{\operatorname{Ext}(\mathcal{T}_+(\v{p}))}$ is the upper bound on the number of vertices of the future thermal cone of $\v{p}$.
\end{lem}
The proper proof of this geometric insight is given in the Appendix~\ref{app:geometric_fancy}. However, an intuition for the upper bound can be found by considering a ball $\mathcal{B}(R, \v{c})$ of radius $R$ centred at the centre $\v{c}$ of the simplex $\Delta_d$. First, for some $R = R_{d-1}$ the ball becomes inscribed, with a single common point with each hyperface of the simplex. Next, we have $R = R_{d-2}$, when the ball touches each hyperedge at exactly one point. For $R_{d-1} < R < R_{d-2}$, we find a $d-1$-dimensional ball on each of the hyperfaces. Replacing the ball with an inscribed inverted simplex yields the desired result.

The analogy drawn between special relativity and Gibbs-preserving matrices introduces a causal structure into the probability simplex $\Delta_d$, suggesting a ``light cone'' for each point within $\Delta_d$ that divides the space into past, incomparable, and future regions. This analogy implies that, much like in special relativity, a specific division of space-time exists. This division separates space-time into future, past, and space-like regions. This allows for the identification of the generating event, which is referred to as the present. Furthermore, there exists a one-to-one correspondence between events and the space-time divisions they induce. This concept is mirrored in the thermal cones for $\beta > 0$: given a specific configuration of the incomparable region along with future and past thermal cones, one can precisely determine the current state of the system [see black dot~$\bullet$ in Fig.\hyperref[Fig:catalysable-thermal-cones-example]{\ref{Fig:catalysable-thermal-cones-example}a-d}]. This is in sharp contrast to the situation with $\beta = 0$, where each division into past, future, and incomparable has a $d!$-fold symmetry, making it impossible to deduce the present state of the system based solely on this division without additional details such as the permutation that arranges the probabilities in non-decreasing order. Interestingly, in certain temperature regimes, the structure of the catalysable future mimics this feature, where the past and future regions meet again~[see Fig.~\hyperref[Fig:catalysable-thermal-cones-example]{\ref{Fig:catalysable-thermal-cones-example}a-d}].

The concept of catalysable regions applies also to other majorisation-based resource theories, such as entanglement and coherence theories, where they proceed from the \cite{grabowecky2018bounds} directly in combination with the tangent vectors. These theories are defined by sets of free operations and free states: local operations and classical communication (LOCC) with separable states in entanglement theory~\cite{Horodecki2009}, and incoherent operations (IO) with incoherent states in coherence theory~\cite{Plenio2014}. In each theory, quantum states can be represented by probability distributions, which are crucial for defining state transformation conditions under free operations. However, in both cases, the partial order that emerges is precisely the opposite of the thermodynamic order in the infinite temperature limit~\cite{nielsen1999conditions,Du2015}. As a result, what constitutes the catalysable thermodynamic past and future in thermodynamics becomes the future and past in entanglement and coherence theories, while the incomparable region remains unchanged. Nevertheless, it is important to highlight that there is no fully equivalent physical situation in either entanglement or coherence theories where thermomajorisation is applicable.

\subsection{Dimensionality bounds for thermal catalyst}

Our second main result concerns general bounds on the catalyst dimension for thermal operations, extending the bounds described in Ref.~\cite{grabowecky2018bounds} for local operations and classical communication. Notably, unlike the entanglement scenario, catalysis is possible even when the main system is a three-level system.

\begin{restatable}[Dimensionality criteria]{thm}{thmCatBeta}\label{thm:cat_beta}
   For two incomparable energy-incoherent states $\v{p} \perp_\beta \v{q}$, we define a maximal interval $\mathcal{L}$ as
    \begin{equation}
        \mathcal{L} = \qty{x\in(0,1)\mid f_{\v{p}}(x) - f_{\v{q}}(x)<0},
    \end{equation}
    and denote $m = \min(\mathcal{L})$ and $n = \max(\mathcal{L})$. Consider a catalyst $\v{r} \in \Delta_k$ with Hamiltonian $H_c$. The transformation can be catalysed by $\v{r}$, that is, $\v{p} \otimes \v{r} \succ_\beta \v{q} \otimes \v{r}$, only if its dimension satisfies the required conditions:
    \begin{align}
        k & > k^* = \frac{\log b}{\log a} + 1, \label{eq:catalyst_dim_bound}
    \end{align}
    with the coefficients $a, b$ defined as 
    \begin{align}
        a & = \min\qty(\frac{s_1\qty(\v{p})}{f'_{\v{p}}(m_-)},\frac{f'_{\v{p}}(n_+)}{s_d\qty(\v{p})})> \max_{v\in(0,1)}\qty(\frac{f'_{\v{r}}(v_-)}{f'_{\v{r}}(v_+)}), \\
        b & = \max_{l\in\mathcal{L}}\qty(\frac{f'_{\v{p}}(l_-)}{f'_{\v{p}}(l_+)}) <  \frac{s_1\qty(\v{r})}{s_k\qty(\v{r})},
    \end{align}
    where for brevity we use $f'(x)\equiv\dv*{f(x)}{x}$ and $f'(x_i)\equiv \eval{f'(x)}_{x = x_i}$. Furthermore, we denote the left and right derivatives as $f'(x_\pm) = \lim_{y\rightarrow x_\pm} f'(y)$.
\end{restatable}
    The full proof of the above theorem is given in Appendix~\ref{app:proofs}. Therein, we present a heuristic proof in which we treat (thermo)majorisation curves as continuous objects, and an exact proof based on the notion of an embedding map~\cite{horodecki2013fundamental}, where thermomajorisation $\succ_\beta$ is equivalent to majorisation $\succ$ in the (embedded) higher-dimensional space. Both approaches agree in the obtained extension of the results presented in~\cite{grabowecky2018bounds}. Crucial insight leading to the results, however, is that catalysability of thermal states is encoded in the slopes of thermomajorisation functions, which may be obtained directly from the relation between the state of the system and the underlying Gibbs distribution. This stands in contrast with the standard understanding that transition from majorisation to thermomajorisation requires generically infinite-dimensional system.

We may further simplify the expressions above to remove optimisation over an interval in favour of a finite set of points.
\begin{lem}[Interval to point optimisation]\label{lem:less_points}
    For a given state $\v{p}$, we define an auxiliary set of indices $\mathcal{L}'\subset\qty{1,\hdots,d}$ as
    \begin{equation}
        \mathcal{L}' = \qty{l\in\qty{1,\hdots,d}\mid  \Gamma_l = \sum_{i=1}^l \gamma_i^\beta, f_{\v{p}}\qty(\Gamma_l) - f_{\v{q}}(\Gamma_l)<0},
    \end{equation}
    where $\v{\gamma}^\beta$ is assumed to be ordered according to the $\beta$-order of $\v{p}$. Then we find that 
    \begin{equation}
        \max_{l'\in\mathcal{L}'} \frac{s_{l'}\qty(\v{p})}{s_{l'+1}\qty(\v{p})}= \max_{l\in\mathcal{L}}\qty(\frac{f'_{\v{p}}(l_-)}{f'_{\v{p}}(l_+)}) = b.
    \end{equation}
\end{lem}
\begin{proof}
    The ratio between left and right derivative is equal to one whenever the function is differentiable. Therefore, the only points where we might find nontrivial values are for the elbows of the thermomajorisation curve $f_{\v{p}}$, described entirely by the set $\mathcal{L}'$.
\end{proof}

Theorem~\ref{thm:cat_beta}  has notable implications for the catalysability of low-dimensional systems. Specifically, we derive the following two corollaries:
\begin{obs}[Thermal non-catalysability for $d=2$]
    Consider two states described by population vectors $\v{p},\,\v{q}\in\Delta_2$ such that $\v{p}\perp_\beta\v{q}$. There exists no state $\v{r}\in\Delta_k$ which can catalyse the transformation.
\end{obs}
\begin{proof}
    A heuristic argument can be given based on the embedding lattice introduced in \cite{de2022geometric} -- any two states $\v{p},\v{q}$ of dimension $d= 2$, for $\beta > 0$ and different $\beta$-orders, have an effective dimension $\tilde{d} = \abs{\qty{\sum_{i\in I}\gamma_i|I\in2^{\qty{1,2}}}}-1=3$, where by $\abs{\cdot}$ we denote the size of the set, and they fall within the scope of Property 2 from \cite{grabowecky2018bounds}.

    For the more formal proof, we start by assuming that $\v{p}$ and $\v{q}$ have different $\beta$-orders. From their incomparability we find that they have to satisfy
    \begin{align}
        q_2 \geq p_1 \frac{\gamma_2}{\gamma_1}, &&
        q_2 + q_1 \frac{\gamma_1-\gamma_2}{\gamma_1} \leq p_1,
    \end{align}
    By simple manipulation, the second inequality is turned into
    \begin{align}
        1 - q_1 + q_1 \frac{\gamma_1-\gamma_2}{\gamma_1} & \leq 1-p_2, \\
        q_1\frac{\gamma_2}{\gamma_1} & \geq p_2.
    \end{align}
    Together with the first inequality we have 
    \begin{align}
        s_1(\v{q}) = \frac{q_2}{\gamma_2} \geq \frac{p_1}{\gamma_1} = s_1(\v{p})&&
        s_2(\v{q}) = \frac{q_1}{\gamma_1} \geq \frac{p_2}{\gamma_2} = s_2(\v{p}). 
    \end{align}
    which contradicts the necessary condition expressed in Eq.~\eqref{eq:nec_cond_beta} for catalysability.
\end{proof}

The more interesting case occurs in dimension $d = 3$:

\begin{obs}[Thermal catalysability in $d=3$]
    There exist pairs of states described by population vectors $\v{p},\,\v{q}\in\Delta_3$ such that $\v{q}\in\mathcal{C}_+(\v{p})$ and a catalyst $\v{r}\in\Delta_k$ such that $\v{p}\otimes\v{r}\succ_\beta\v{q}\otimes\v{r}$.
\end{obs}
\begin{proof}
 An explicit example can be given by considering a pair of incomparable vectors $\v{p} = \qty(0.42,0.51,0.07)$ and $\v{q} = \qty(0.52,0.13,0.05)$, both with energy spectrum $E_0 =0$, $E_1=1$ and $E_2=2$ and inverse temperature $\beta = 0.2$. Then, a 2-dimensional catalyst, described by a state $\v{r} = (0.55,0.45)$ with trivial Hamiltonian $H_{\v c} = 0$, catalyses the transformation, since $\v{p}\otimes\v{r}\succ_\beta\v{q}\otimes\v{r}$.
More generally, it is simple to see that as soon as $\beta >0$, we find that for a generic state $\v{p}$, the catalysable future is non-empty $\mathcal{C}_+(\v{p})\neq\emptyset$. 
   \end{proof}

Furthermore, using Theorem \ref{thm:cat_beta} we can formulate certain bounds on qubit catalysts, limiting the corresponding probability vector $\v{r}$ to a specific interval.

\begin{cor}[Trivial qubit catalysts]\label{cor:qubit_cat_beta}
    Consider a pair of incomparable states $\v{p},\,\v{q}\in\Delta_d$ such that $\v{q}\in\mathcal{C}_+(\v p)$. A qubit catalyst in a state $\v{r} = (1-t,t)$ with $t \leq 1/2$ and described by a trivial Hamiltonian $H_c = 0$ can catalyse the transformation, $\v{p}\otimes\v{r} \succ \v{q}\otimes\v{r}$ only if
    \begin{equation}
        \frac{1}{1+a} \leq t \leq \frac{1}{1+b}
    \end{equation}
    with $a,b$ defined as in Theorem \ref{thm:cat_beta}.
\end{cor}

In contrast to \cite{grabowecky2018bounds}, the provided bound on qubit catalyst is expressed in terms of coefficients $a, b$ which are implicitly functions of both the underlying probability vector $\v{p}$ and the underlying energy level structure given by the Gibbs state $\v{\gamma}$. Moreover, similar bounds can be given for qubit catalysts with nontrivial energy level structure.

\begin{cor}[Nontrivial qubit catalysts]\label{cor:qubit_cat_beta_nont}
    Consider a pair of incomparable states $\v     {p},\,\v{q}\in\Delta_d$ such that $\v{q}\in\mathcal{C}_+(\v{p})$ and a qubit catalyst in a state $\v{r} = (1-t,t)$ with Hamiltonian such that the Gibbs state is given by $\v{\gamma}_{\v{r}} = \qty(1-\gamma_{\v{r}},\,\gamma_{\v{r}})$. For $t \leq \gamma_{\v{r}}$ it can catalyse the process only if
    \begin{equation}
        \frac{\gamma_{\v{r}} }{a(1-\gamma_{\v{r}})+\gamma_{\v{r}} } \leq t \leq \frac{\gamma_{\v{r}} }{b(1-\gamma_{\v{r}})+\gamma_{\v{r}}};
    \end{equation}
    for $t \geq \gamma_{\v{r}}$ it can catalyse the transformation only if
    \begin{equation}
        \frac{1-\gamma_{\v{r}} }{b\gamma_{\v{r}} + 1 - \gamma_{\v{r}}} \leq 1-t\leq \frac{1-\gamma_{\v{r}} }{a\gamma_{\v{r}} + 1 - \gamma_{\v{r}}},
    \end{equation}
    where in the above we take $a,\,b$ from Theorem \ref{thm:cat_beta}.
\end{cor}

It is worth noting that the two bounds above are symmetric with simultaneous replacement of $t\leftrightarrow1-t$ and $\gamma_{\v{r}}\leftrightarrow1-\gamma_{\v{r}}$ and thus, in a certain sense, the allowable region for qubit catalysts retains the symmetry with respect to the Gibbs state.
\begin{figure*}
    \centering
    \includegraphics{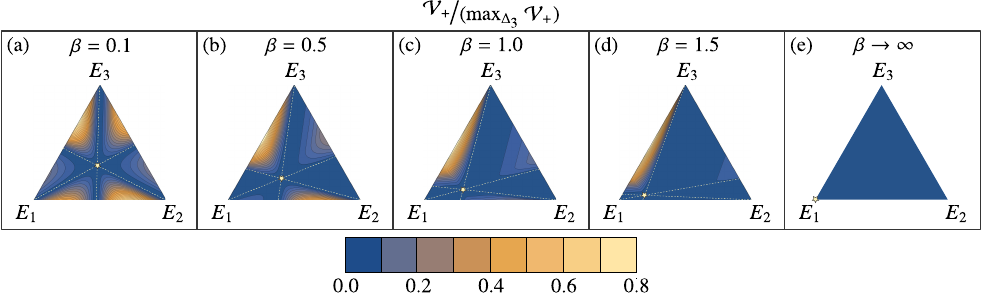}
    \caption{\textbf{Isovolumetric curves for $d=3$.} 
For a three-level system described by an equidistant energy spectrum with $E_1 = 0$, $E_2 = 1$, and $E_3 = 2$, we plot the volume of the catalysable future for several regimes of $\beta$. The lines indicate states with equal volumes, and the colours represent their magnitudes as a fraction of the maximal volume, $\mathcal{V}_+/\qty(\max_{\Delta_3}\mathcal{V}_+)$. The yellow star depicts the thermal state, and the dashed line indicates the boundaries between the six possible $\beta$-orderings, dividing the simplex into six chambers. Note that at the boundaries $\mathcal{V}_+$ goes to zero, while maxima are located at the outer edges of each chamber.}
    \label{F:iso-volumes}
\end{figure*}

\subsection{Quantifying catalytic advantages \label{sec:volume}} 

The characterization of the catalysable regions naturally enables the quantification of the catalytic advantage through their respective relative volumes. Given a state $\v{p}$, we define
\begin{equation}\label{Eq:volume-past-future-cat}
    \mathcal{V}_{i}(\v p) = \frac{\operatorname{Vol}[\mathcal{C}_i(\v p)]}{\operatorname{Vol}(\Delta_d)} \quad \text{for} \quad i\in \{+,-\},
\end{equation}
as the volumes of the catalysable past and future, where $\operatorname{Vol}(\cdot)$ denotes the Euclidean metric volume. Both volumes also help us to understand the behaviour of the catalytic set for a given choice of state and temperature.

We begin our analysis of the volumes of catalysable thermal cones by presenting general properties that can be directly read from our results. The starting point is a straightforward result that follows directly from Theorem~\ref{Thm:catalysable-set}, namely
\begin{cor}[Zero volume for non-full rank states]
\label{Cor:catalysable-past-volume}
the catalysable past thermal cone of a non-full rank state has zero volume.
\end{cor} 
\begin{proof}
Without loss of generality, consider a non-full rank state $\v{p} = (p_1, ..., p_{d-1},0)$. Applying Eq.~\eqref{eq_thermaltangentvectors} yields $t^{(d,\v \pi)}_i(\v p) = \delta_{\v{\pi}(1),i}$, for all $\v \pi \in \mathcal{S}_d$. Consequently, the incomparable region is given by all points in the interior of the probability simplex, except those that are in the future of $\v{p}$. Then, all the points of the past will be located at the edge. Therefore according to Theorem~\ref{Thm:catalysable-set} there is no catalysable past and its volume is zero.
\end{proof}
The above result also enables us to state a general behaviour for sharp states $\v{v}_k$ with $(\v{v}_k)_j = \delta_{jk}$:
\begin{cor}[Zero volume for sharp states]
\label{Cor:catalysable-future-sharp}
the catalysable future thermal cone of a sharp state $\v{v}_k$ has zero volume.
\end{cor} 
\begin{proof}We proceed by showing that the extreme points of the catalysable set coincide with the extreme points of the future thermal cone of the sharp state. First, we
consider the slope vector $\v{s}(\v{v}_k)$ of a sharp state $\v{v}_k$. We readily determine that $s_j(\v{v}_k) = \frac{\delta_{j1}}{\gamma_k}$. From this, it is straightforward to conclude that the corresponding $d$-th tangent vectors are represented by sharp states, i.e., $\v{t}^{(d,\v \pi)}(\v{v}_k) = \v{v}_{\pi(1)}$. The components of the first ones are expressed as $t^{(1,\v \pi)}_i(\v{v}_k) = \gamma_i/\gamma_k$ for $i<d$. Second, upon defining a projection operator $\hat{P}$ that acts on a probability vector, we ensure that the thermomajorisation function at each elbow is defined to be
\begin{equation} \label{eq:projection_operation}
    f_{\hat{P}\v{t}}(\Gamma_i) = \max\qty[1,f_{\hat{P}\v{t}}(\Gamma_i)]\quad , \quad \forall\,\Gamma_i = \sum_{j=1}^i \gamma_{\pi(j)}.
\end{equation}
Using the above projection onto the  simplex $\Delta_d$ one finds that either $\hat{P}\v{t}^{(1,\v \pi)}(\v{v}_k) = \v{v}_k$ or $\hat{P}\v{t}^{(1,\v \pi)}(\v{v}_k)$ is an extreme point of the future thermal cone $\mathcal{T}_+(\v{v}_k)$. Thus, we find that $\mathcal{T}_+(\v{v}_k) = T_1(\v{v}_k)$ and, in consequence, $T_1(\v{v}_k)\cap T_d(\v{v}_k) \setminus \mathcal{T}_+(\v{v}_k) = \emptyset$, which completes the proof.
\end{proof}

Note that there is an interplay between the volumes of the catalysable future and past, and the volume of the incomparable thermal region. Understanding the general behaviour of their volumes is complex as it depends on the state under consideration. Nevertheless, a significant insight can be obtained by resorting to the fact that the state of non-full rank with the smallest future thermal cone is known~\cite{de2022geometric}. For this state, the following result can be proven:
\begin{cor}[Non-catalysable state]
\label{Cor:catalysable-sharp}
The non-full rank state with the smallest future thermal cone
\begin{equation}
\label{eq_nonfullrankstatethermalinc}
\v g = \frac{1}{Z}_{\v g}\left(e^{-\beta E_1}, ...,e^{-\beta E_{d-1}},0 \right) \quad \text{where} \quad Z_{\v g} = \sum_{i=1}^{d-1}e^{-\beta E_i},
\end{equation}
cannot be catalysed, i.e., the volume of its catalysable future is zero.
\end{cor} 
\begin{proof}
Given the assumption that if $i>j\Rightarrow E_i<E_j$, we conclude that  $\gamma_{i}\geq \gamma_j$. From this, we observe that that either $\v{t}^{(d,\v \pi)}(\v{g}) = \v{g}$ or $\v{t}^{(d,\v \pi)}(\v{g})$ is an extreme point of the future thermal cone $\mathcal{T}_+(\v{g})$. Therefore, we find that $\mathcal{T}_+(\v{g}) = T_d(\v{g})$, leading to the conclusion that $T_1(\v{g})\cap T_d(\v{g}) \setminus \mathcal{T}_+(\v{g}) = \emptyset$.
\end{proof}

The analysis and discussion of the volumes has proceeded without explicit computation. Several algorithms are known for calculating the volumes of convex polytopes~\cite{iwata1962, buller1998}. More specifically, it is observed that the catalysable future is given by the intersection of two simplices, which reduces the complexity of volume calculation thanks to the reduction of the number of vertices (see Lemma~\ref{lem:less_points}) in comparison to the full thermal future cone. 

One might ask about the states in $\Delta_d$ that possess the highest volumes of catalysable future. Answering this question generally is highly non-trivial. The challenge arises because the volumes depend significantly on the initial state, and computing them for higher dimensions is not straightforward. However, by resorting to low-dimensional systems, we can still gain some insights into how catalysability behaves as a function of inverse temperature. In this regard, we illustrate the profile of the catalysable future's volume as a function of $\beta$ in Fig.~\ref{F:iso-volumes} for a three-level system. It is observed that states with the highest volumes are concentrated near the edges of the probability simplex, and the volume of the catalysable future tends to be higher for states that are not of full rank.
Interestingly, in the extreme cases of $\beta = 0$ and $\beta \to \infty$, catalysis has no effect. 


\section{Applications}
\subsection{Entanglement generation under thermal operations with catalyst \label{Sec:applications-1}} 

\begin{figure*}
    \centering
    \includegraphics{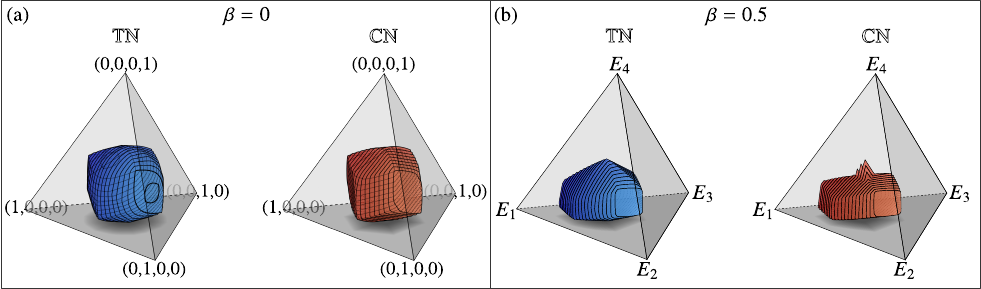}
    \caption{\textbf{Thermally non-entanglable states with (and without) catalyst}. 
    Two-qubit states that cannot become entangled under thermal operations are illustrated for two scenarios: without a catalyst (shown in the left panel in blue) and with a catalyst (shown in the right panel in red), for (a) $\beta =0$ and (b) $\beta = 0.5$}
    \label{F:entanglable-catalyst}
\end{figure*}

As a direct application of our results, we demonstrate the catalytic advantages of generating entanglement from separable bipartite states under thermal operations. Throughout this section, we will focus on a composite system consisting in two 2-level systems with identical energy levels $H = \operatorname{diag}(0, \delta)$. For this reason, the total Hamiltonian of the system, $H_{AB} = H\otimes\mathbbm{1} + \mathbbm{1}\otimes H$ has a degenerate energy subspace corresponding to $E_{01} = E_{10} = \delta$, and without loss of generality we may set $\delta = 1$. Within this subspace any unitary operation remains Gibbs-preserving, which allows for generation of entanglement while simultaneously respecting the laws of thermodynamics.

Recently, the necessary and sufficient conditions for producing bipartite qubit entanglement from an initially separable states via thermal operations were addressed in Ref.~\cite{de2024ent}. States that cannot get entangled under thermal operations are said to belong to the \emph{thermally non-entanglable set}, mathematically denoted by $\mathbb{TN}$. Conversely, states that can get entangled under thermal operations lie in the \emph{entanglable set} $\mathbb{TE}$.
By using the results on the catalysable future region $\mathcal{C}_+$, we may now answer the natural question whether borrowing catalyst allow more states to get entangled. Defining the \emph{catalytically non-entanglable set} $\mathbb{CN}$ as the set of states that cannot be entangled even with the aid of catalysis, the question can be reframed -- \emph{is the set $\mathbb{CN}$ smaller than $\mathbb{TN}$, or in terms of volumes, $\mathcal{V}(\mathbb{CN})/\mathcal{V}(\mathbb{TN})\overset{?}{<} 1$ for a given $\beta$?}

To begin with, we focus on the infinite temperature case $\beta = 0$. Remarkably, in this regime, the thermally non-entanglable set $\mathbb{TN}$ reaches its maximum volume and can be analytically constructed. In order to investigate the catalytically non-entanglable set $\mathbb{CN}$, we start by using Theorem~\ref{Thm:catalysable-set}, where we construct the catalysable future for two qubits and verify whether it lies within the thermally non-entanglable set. Next, by employing a numerical bisection-based algorithm inspired by vacuum-packaging, we generate an approximation of the boundary $\partial\mathbb{CN}$. This allows us to estimate the ratio of the volumes $\mathcal{V}(\mathbb{CN})/\mathcal{V}(\mathbb{TN})\approx 0.88$. In Fig.~\hyperref[F:entanglable-catalyst]{\ref{F:entanglable-catalyst}a}, we show the thermally non-entanglable set without (left panel) and with (right panel) a catalyst.

It is interesting to note that for $\beta > 0$, although the set $\mathbb{TN}$ remains convex as verified by numerical approximations, the set $\mathbb{CN}$ becomes non-convex~(see Fig.~\hyperref[F:entanglable-catalyst]{\ref{F:entanglable-catalyst}b} for an example considering $\beta =0.5$). In fact, by reverse-engineering the numerical approximations, we find that $\mathbb{CN}$ is a set sum of a certain convex set $\mathbb{CN}_0$ and a future of a certain distinguished subspace-thermalised state, $\mathcal{\v{p}_*}$. First, concerning the distinguished state, we have

\begin{prop}[Partially thermalised state \& non-entanglability] \label{prop:peculiar_point}
    The future thermal cone $\mathcal{T}_+(\v{p}^*)$ of the partially thermalised state
    \begin{equation}
        \v{p}^* = \frac{1}{4+2\cosh(\beta )}\qty(e^{\beta },\,1,\,1,\,2+e^{-\beta })
    \end{equation}
    is fully contained in the catalytically non-entanglable set, $\mathcal{T}_+(\v{p}^*)\subset\mathbb{CN}$.
\end{prop}
We give the entire proof of this fact in Appendix \ref{app:peculiar_point}. Furthermore, based on numerics we put forward the following 
\begin{conj}[Non-entangable decomposition]
    The catalytically non-entanglable set can be decomposed into
    \begin{equation}
        \mathbb{CT} = \mathbb{CT}_0\cup\mathcal{T}_+(\v{p}^*),
    \end{equation}
    where
    \begin{equation}
        \mathbb{CT}_0 = \mathbb{CT} \setminus\qty[\mathcal{T}_+(\v{p}^*) \setminus\mathcal{T}_+(\v{p}^{**})],
    \end{equation}
    with $\v{p}^{**} = \frac{1}{4+2\cosh(\beta)}\qty(e^{\beta },\,3,\,1,\,e^{-\beta })$.
\end{conj}
The conjecture is based on numerical approximations with high precision, and we believe it to hold true. Furthermore, it can be partially supported by noting that for a ball $\mathcal{B}(\v{p}^*,\,\delta)$ of radius $\delta$ around the state $\v{p}^*$ any state $\v{q}\in\qty[\mathcal{B}(\v{p}^*,\,\delta)\setminus\mathcal{T}_+(\v{p}^*)]$ does not belong to $\mathbb{CN}$ by simple arguments of either not belonging to $\mathbb{TN}$ or by the fact that $\mathcal{C}_+\qty(\v{q})\setminus\mathbb{TN}\neq \emptyset$.

Using the introduced numerical approximations we derived the dependence of the volumes $\mathcal{V}(\mathbb{TN})$ and $\mathcal{V}(\mathbb{CN})$ for both thermal and catalytically non-entanglable sets and their ratio as a function of $ \beta$, demonstrated in Fig~\ref{fig:volume_ratio}. 

\begin{figure}[t]
    \centering
    \includegraphics{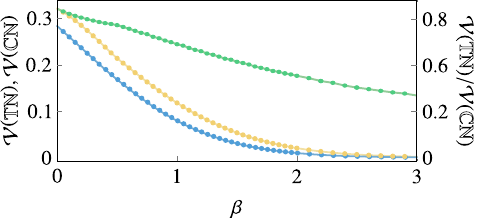}
    \caption{\textbf{Thermally non-entanglable volume with(out) catalyst.} Numerical estimates for the volumes $\mathcal{V}(\mathbb{CN})$ (blue line), $\mathcal{V}(\mathbb{TN})$ (orange line) and their ratio (green line), with the axis on the left for volumes and on the right for the ration. The character of the curves shows that, even though the volumes of both thermally and catalytically non-entanglable sets both go to zero together with $\beta$ going to infinity, the ratio of the volumes also diminishes, pointing to an increasing catalytic advantage in entanglement generation.}
    \label{fig:volume_ratio}
\end{figure}

\subsection{Optimal cooling with a catalyst}\label{sec:cool}

As a second application of our findings, we explore the task of cooling a three-level system with the aid of a catalyst. Our setup comprises a three-level system, a thermal environment, and a catalyst. We assume that the composite system forms a closed system and evolves unitarily under an energy-preserving interaction (as discussed in Sec. \ref{Sec:framework}). Consequently, the energy exchange in the process is account for as heat, which can be determined by calculating the energy difference between the initial and final states after the interaction. Thus, our figure of merit is defined as follows
\begin{align}
    Q_c(\rho) = \underset{\sigma \, \in\,  \T_{+}(\rho)}{\min} \:\: &\tr[H(\sigma-\rho)].
\end{align}
where $\rho$ and $\sigma$ are energy-incoherent states. In other words, we ask for the optimal thermal operation which cools down the system $Q_c(\rho) < 0$, while heating up the environment. As discussed in Sec.~\ref{Sec:framework}, the set of achievable states under thermal operation is convex and one can prove that the optimal cooling (when it happens) occurs when the initial state is mapped to one of the extreme points of its future thermal cone. Now one can ask, what is the optimal cooling that can be achieved when we allow a strict catalyst? This question can be answered by employing Theorem~\ref{Thm:catalysable-set} and evaluating the extreme point of the catalysable set. Since answering such a question in a general manner is highly non-trivial due to the state-dependent character of the problem, we can assume that we have access to the initial state preparation of the system and use our results to find the optimal cooling.

Consider, for instance, an initial state denoted by $\v{p} = (0.1, 0.2, 0.7)$, characterised by a Hamiltonian with an equidistant spectrum: $E_0 = 0$, $E_1 = 1$, and $E_2 = 2$. Additionally, let us assume that the system can interact with a thermal environment prepared at an inverse temperature of $\beta = 0.2$. This state has $\beta$-ordering $\v{\pi}_{\v{p}} = (3,2,1)$. To effectively cool the system, the optimal strategy is to send it to the extreme point of its future thermal cone with $\beta$-ordering $\v{\pi}^{\star} = (1,2,3)$. Consequently, the initial state $\v{p}$ transforms to $\v{p}^{\v{\pi}^{\star}} \approx (0.78, 0.15, 0.07)$. This transformation results in a heat exchange of approximately $Q_c \approx -1.3135$, effectively cooling the system down.

Allowing a strict catalyst to aid the process, as given in Theorem~\ref{Thm:catalysable-set}, we observe a shift in the extreme point with $\beta$-ordering $\v{\pi}^{\star} = (1,2,3)$. This shift allows for a more ``substantial cooling''. According to Corollary~\ref{Cor:extreme-points-catalysable}, our results suggest that the initial state $\v{p}$ could now be transformed to $\v{p}^{\star} \approx (0.85, 0.08, 0.07)$ (see Fig.~\ref{fig:cooling-cone}). In this process, the heat exchange is approximately $Q^{\star}_{c} \approx -1.38$, surpassing $|Q_c|$ in magnitude. This simple example illustrates the fundamental limits for cooling a quantum system under thermal operations.
\begin{figure}[t]
    \centering
    \includegraphics{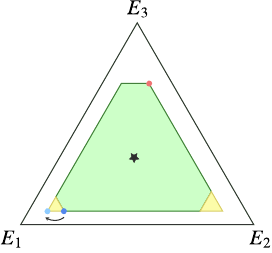}
    \caption{\textbf{Catalytic cooling}. For a three-level system with a population given by $\v{p} = (0.1, 0.2, 0.7)$, represented by a red dot, and energy spectrum $E_1 = 0, E_2 = 1, E_3 = 2$ for $\beta = 0$, the optimal cooling (without a catalyst) consists of mapping $\v{p}$ to the extreme point of its future thermal cone with $\beta$-ordering (1, 2, 3) (dark blue point). When a catalyst is allowed, its future thermal cone expands, and the optimal cooling is achieved by sending $\v{p}$ to the extreme point of the catalysable set with the same $\beta$-ordering (light blue dot). The arrow from the dark blue point simply indicates that a catalyst allows for bringing the initial state closer to the ground state compared to the case without a catalyst.}
    \label{fig:cooling-cone}
\end{figure}

A more general statement can be made by considering a $d$-dimensional system described by an equidistant Hamiltonian $H=\sum_{n=0}^{d-1}n\ketbra{n+1}{n+1}$, with $n\in \qty{0,d-1}$, in contact with a thermal bath at an inverse temperature $\beta$. Assuming that the system is prepared in a thermal state at some hotter inverse temperature $\beta_{\text{h}}<\beta$, its $\beta$-ordering is $(d,d-1,\hdots,1)$ and its slope vector has entries given by $s(\v{p})i = \exp[\qty(\beta - \beta_{\text{h}})E_{d-i}]/Z_\text{h} Z$, where $Z_{\text{h}} = \sum_{i=0}^{d-1}e^{-\beta_{\text{h}} E_i}$. The optimal cooling is generically achieved by transforming the state $\v{p}$ into a state $\v{q}$ with $\beta$-order $(1,2,\hdots,d)$. Observing that for each $j \in \{0, d-1\}$, there exists an $m_j \in \{0, d-1\}$ such that the inequality
\begin{equation}\label{eq:lowerupper_gamma_bounds}
\sum_{i=1}^{m_j} \gamma_{d-i+1} \leq \sum_{i=1}^j\gamma_i \leq \sum_{i=1}^{m_j+1} \gamma_{d-i+1}
\end{equation}
holds, we can state that $\sum_{i=1}^{m_1} p_{d-i+1} \leq q_1 \leq \sum_{i=1}^{m_1+1} p_{d-i+1}$ is valid for the population of $\v{q}$ corresponding to the ground state. On the other hand, we can consider the catalytic transformation $\v{p}\otimes\v{c}\rightarrow\v{q}_{c}\otimes\v{c}$, where $\v{q}_{c}$ represents the optimal target state achieved with the catalyst. Taking into account Theorem~\ref{Thm:catalysable-set}, we find that 
\begin{equation}
    (\v{q}_c)_1 \leq s(\v{p})_1 \gamma_1 = \frac{1}{Z_\text{h}} e^{\qty(\beta - \beta_{\text{h}})E_d - \beta E_1}.
\end{equation}
Given our assumption that the system Hamiltonian is described by an equidistant spectrum, we ensure that $(\v{q}_c)_1\geq q_1$ when the following inequality holds:
\begin{equation}\label{Eq:critical-cooling-1}
    \frac{1-e^{-\beta d}}{1-e^{-\beta}}\frac{1-e^{-(m_j+1)\beta_\text{h}}}{1 - e^{-\beta_\text{h}}} \leq  e^{(\beta-\beta_\text{h})(d-1)}.
\end{equation}
While there is no closed-form solution to the above inequality, one can still numerically evaluate critical temperatures $\beta^{\downarrow}_{\text{h}}$, for which the catalytic advantage is guaranteed. Conversely, if
\begin{equation}\label{Eq:critical-cooling-2}
    e^{(\beta-\beta_\text{h})(d-1)}\leq \frac{1-e^{-\beta d}}{1-e^{-\beta}}\frac{1-e^{-m_j\beta_\text{h}}}{1 - e^{-\beta_\text{h}}},
\end{equation}
we obtain a no-go result indicating that no catalytic advantage for cooling starting from a hot Gibbs state can be achieved above a certain inverse temperature $\beta^{\uparrow}_{\text{h}}$.

For very high temperatures, i.e., when $
\beta_\text{h}<\beta \ll 1$, we can expand both $\beta$ and $\beta_\text{h}$ in Eqs.~\eqref{Eq:critical-cooling-1}-\eqref{Eq:critical-cooling-2} to linear order to obtain approximate critical temperatures as
\begin{equation}
    \begin{aligned}
        \beta^{\downarrow}_{\text{h}} & = \frac{d \qty[3 \beta  (d-1)-2] (m_j+1)+2}{2 d-m_j-2}, \\
        \beta^{\uparrow}_{\text{h}} & = \frac{d \qty[3 \beta  (d-1)-2] m_j+2}{2 d-m_j-1}.
    \end{aligned}
\end{equation}
Consequently, catalytic advantages in cooling the initial system are guaranteed for $\beta_h \leq \beta^{\downarrow}{\text{h}}$. Conversely, for $\beta_h \leq \beta^{\uparrow}{\text{h}}$, we encounter a no-go result indicating that such cooling cannot occur.

Finally, let us mention that there have been several studies on the problem of cooling, either using the resource-theoretic approach to thermodynamics or under a different umbrella of quantum thermodynamics~\cite{Masanes2017, Scharlau2018, Clivaz2019, Taranto2023}. In particular, Ref.~\cite{Wilming2017} uses the approach of catalytic thermal operations for the task of cooling. However, their focus is on providing sufficient and necessary conditions on the amount of resources needed to cool a system close to its ground state. Additionally, the catalyst is allowed to return $\epsilon$-away from its initial state. Therefore, the results presented in this section complement the previous ones by finding the optimal bounds on the heat exchange, which leads to cooling the initial system.  


\section{Summary and Discussion} \label{sec:outlook}

In this paper, we explore the fundamental limits on state achievability with the aid of a strict catalyst under a thermal process. We propose an approach to address this question by constructing the sets of states that can be achieved from a given energy-incoherent state under catalytic thermal operations, as well as the sets of states that can, through catalytic thermal operations, achieve the initial state. These regions naturally highlight the advantages of using a strict catalyst in general thermodynamic processes, whose only assumption is energy conservation. Our construction enables us to characterise the catalysable past and future for main systems of arbitrary dimensions. This analytical approach is based solely on thermomajorisation relations and calculation of slopes of the resulting thermomajorisation curves, providing a way of verifying non-catalysability in finite number of steps. Furthermore, we established general bounds on the catalyst dimension for thermal operations and demonstrated that catalysis is effective under thermal operations even for systems of dimensionality as small as $d=3$. The catalytic advantages were quantified by the volume of the catalysable future, with a detailed discussion on its behaviour as a function of temperature. Moreover, we applied our findings to address two practical questions: the generation of entanglement under thermal operations and the optimal cooling of a quantum system. In the first application, we showed that thermal processes, which were once incapable of generating entanglement, allows its generation when a strict catalyst is employed. While in the second, we discussed the ultimate limits of cooling aided by catalysts.

To address this problem, we overcome several challenges. Firstly, there is no efficient method for characterizing the set of allowed state transformations under catalytic thermal operations. A comprehensive analysis for inexact catalysis (where the catalyst is allowed to be returned with some minor error) was performed in~\cite{Lipka-Bartosik2021}; while, in Ref.~\cite{son2022catalysis}, the authors solved the characterization problem for three-dimensional systems and for a subset of initial states of generic dimensions in the context of catalysis in elementary thermal operations. Secondly, keeping track of the the set of achievable state after ``decoupling'' the catalyst  without specifying the precise form of the catalyst state required by the transformation is a highly difficult mathematical problem. Approaching from the perspective of lattice theory, this problem can be framed by asking \emph{(i)} how to characterise the higher-dimensional lattice for composite systems (main system plus catalyst) and \emph{(ii)} how to project back to the native lattice of the main system while preserving the catalytic effects in the native lattice of the initial state. For small-dimensional systems (such as $d=3$), the entire convex polytope $\v{p} \otimes \v{c}$ can be constructed, and intersections between such a polytope and plane equations enforcing the catalyst recovery condition can be numerically identified. However, this method does not scale well and depends heavily on the initial state of the system and the catalyst. In contrast, our approach requires only constructing a set of quasi-probability vectors or determining the extremal points through thermomajorisation relations.

There are many possibilities for extending our results. For instance, they can be leveraged to study catalysis within the framework of elementary thermal operations. Note that the characterization of the catalysable past and future relies entirely on the concept of tangent vectors, which tightly thermomajorise $\v p$. On the other hand, the extreme points of set of states achievable under elementary thermal operations are known to be $\beta$-swaps that tightly thermomajorise the initial distribution. Consequently, one could explore how to adapt the technical aspects of our findings to this specific context. In a similar spirit, we could relax the condition requiring the catalyst to be returned unperturbed and allow for some small error. A natural quantifier for this error is the trace distance, and the problem could be reframed by asking for the set of states achievable with a given catalyst up to this specified amount of error~\cite{Lipka-Bartosik2021}. 

It is tempting to suggest that the results obtained in this work, as an extension of prior results of Grabowecky and Gour~\cite{grabowecky2018bounds}, may point to even more general statements. Noting that majorisation and thermomajorisation curves are, in fact, totally ordered set of monotones, one may imagine the following statement: given two states, $X$ and $Y$, are such that the first and last monotones for $X$ are greater than that of $Y$, then the transformation is catalysable, thus providing a method for constructing the catalysable past and future for more general resource theories. This, however, is beyond the scope of this paper.

Finally, our results imply that one can potentially (for a fixed finite dimension) replace an infinity family of second laws by a finite set of conditions. As it can be checked, within the catalysable region the second laws are fulfilled and as soon as we get out of it, they are broken. So, an ambitious question is whether our results can be reformulated (or adapted), so that given $\v p$ and $\v q$, one can check a finite number of conditions to verify whether $\v p$ can catalytically thermomajorises $\v q$. Note that, our results shows that for every $\v p$, we have a method to calculate the extended future thermal cone. For each $\beta$-order, this cone is defined by an extreme point of the catalysable set. While membership in this region does not guarantee the existence of a catalyst, non-membership guarantees its non-existence.

\begin{acknowledgments}
We thank Karol Życzkowski for asking a stimulating question at the beginning of the authors' PhD, which led to this work, and for his fruitful comments on the first version of this manuscript. We also thank Kamil Korzekwa for valuable discussions that significantly contributed to the execution of this project, as well as for his insightful comments during the final stages of preparation of this manuscript.

We acknowledge the financial support from Danish National Research Foundation's bigQ grant (DNRF 142) is gratefully acknowledged. JCz acknowledges financial support by NCN PRELUDIUM BIS no. DEC-2019/35/O/ST2/01049 and ``Research support module'' as part of the ``Excellence Initiative -- Research University'' program at the Jagiellonian University in Kraków. AOJ acknowledges financial support from VILLUM FONDEN through a research grant (40864) and EU Horizon Europe project QSNP (grant agreement no. 101114043).
\end{acknowledgments}

\appendix

\section{Reframing entanglement catalysis bounds for thermal operations}

We start by revisiting the results presented in Ref.~\cite{grabowecky2018bounds} and adapt them to the context of thermal operations at the infinite temperature limit. First, we establish a simple necessary condition: For a pair of vectors $\v{p}$ and $\v{q}$, a catalyst $\v{r}$ such that $\v{p} \otimes \v{r} \succ \v{q} \otimes \v{r}$ exists only if $p^\downarrow_1 < q^\downarrow_1$ and $p^\downarrow_d > q^\downarrow_d$. Based on this assumption, we then restate the main theorem from the aforementioned work. 
\begin{restatable}[Theorem 1 of Ref.~\cite{grabowecky2018bounds}]{thm}{thmCatInfty}\label{thm:cat_infty}
    Consider two states described by population vectors $\v{p},\,\v{q}\in\Delta_d$ which are incomparable, $\v p\perp\v q$ (ie. neither $\v  p\prec\v q$ nor $\v q \prec \v p$) and $p^\downarrow_1 > q^\downarrow_1$ and $p^\downarrow_d<q^\downarrow_d$. We define a set of indices $\mathcal{L}$ as:
    \begin{equation}
        \mathcal{L} = \qty{l\in\qty{1,2,...,d}\mid \sum_{i=1}^l p^\downarrow_i - q^\downarrow_i<0}
    \end{equation}
    and take $m = \min(\mathcal{L})$ and $n = \max(\mathcal{L})$. The catalyst $\v{r}\in\Delta_k$ can catalyse the transformation, ie. $\v p\otimes\v r\succ\v q \otimes \v r$, if and only if its dimensionality satisfies
    \begin{align}
        k & > k^* = \frac{\log b}{\log a} + 1 
    \end{align}
    with the coefficients $a, b$ defined as 
    \begin{align}
        a & = \min\qty(\frac{p^\downarrow_1}{p^\downarrow_m},\frac{p^\downarrow_{n+1}}{p^\downarrow_d}) > \max_{v\in\qty{1,2,\hdots,k-1}}\qty(\frac{r^\downarrow_v}{r^\downarrow_{v+1}})\label{eq:a_inft_ver}\\
        b & = \max_{l\in\mathcal{L}}\qty(\frac{p^\downarrow_l}{p^\downarrow_{l+1}}) < \frac{r^\downarrow_1}{r^\downarrow_k} \label{eq:b_inft_ver}
    \end{align}
    and provide bounds for the entries of the catalyst state $\v{r}$.
\end{restatable}
Note that the bound depends explicitly only on the entries of the initial state $\v{p}$, with the target state $\v{q}$ entering the picture indirectly via the set $\mathcal{L}$, which defines the range of the entries taken into consideration for optimisation purposes. In a more recent work bounds using both $\v{p}$ and $\v{q}$ entries explicitly have been proposed~\cite{Getal21}; however, the bounds introduced therein are necessarily weaker than the ones introduced in~\cite{grabowecky2018bounds}, and therefore redundant.

We also recall the resulting limitations of the catalytic processes for $\beta = 0$:
\begin{obs}
    For incomparable states of dimension \mbox{$d\leq 3$} catalysis with incoherent catalyst is not possible for noisy operations.
\end{obs}
The proof is given in Ref.~\cite{Plenio1999}. Moreover, in Ref.~\cite{grabowecky2018bounds}, we find two further corollaries of interest.
\begin{cor}
    Whenever $p^\downarrow_1 = p^\downarrow_m$ or $p^\downarrow_{n+1} = p^\downarrow_d$, catalysis is impossible.
\end{cor}
\begin{proof}
    In either case we have $a = 1$, and hence by Theorem \ref{thm:cat_infty} one finds that $k^*\rightarrow\infty$.
\end{proof}
Finally, it is important to note that the results immediately define an admissible region for an incoherent qubit catalyst.
\begin{cor}\label{cor:qubit_cat_infty}
    Consider a pair of incomparable states $\v{p},\,\v{q}\in\Delta_d$ such that $\v{p}\perp\v{q}$, $p_1^\downarrow < q_1^\downarrow$ and $p_d^\downarrow > q_d^\downarrow$. A qubit catalyst in a state $\v{r} = (1-t,t)$ with $t \leq 1/2$ can catalyse the process, $\v{p}\otimes\v{r} \succ \v{q}\otimes\v{r}$ only if
    \begin{equation}
        \frac{1}{1+a} \leq t \leq \frac{1}{1+b}
    \end{equation}
    with $a,b$ defined as in Theorem \ref{thm:cat_infty}.
\end{cor}
It has to be stressed, however, that it is only a necessary condition. As has been already noticed in \cite{grabowecky2018bounds}, not only the actual catalysis happens for a smaller sub-interval, but may also happen in several disjoint sub-intervals.

\section{Proof of Theorem \ref{thm:cat_beta}}
\label{app:proofs}

Below we provide two approaches to extend the Theorem~\ref{thm:cat_infty} to the thermal operations. First, we focus on a heuristic approach by treating the majorisation and thermomajorisation curves as continuous functions, which allows us to reformulate the expressions in terms of slopes instead of probabilities. The second, more formal approach, takes advantage of the embedding scheme introduced in \cite{horodecki2013fundamental}, which relates thermomajorisation relation with a majorisation relation in an extended space of potentially infinite dimension.

\subsection{Heuristic approach via derivatives}

Let us consider a state $\v{p}$ and the corresponding majorisation curve $f_{\v{p}}(x)$. In what follows, for brevity of notation, we will use primed notation $f'(x) \equiv \dv*{f(x)}{x}$ for derivative  of a function and $f'(x_i)\equiv \eval{f'(x)}_{x=x_i}$. For $\beta=0$, we find a simple relation between probabilities and derivatives of the majorisation curve, namely
\begin{equation}\label{Eq:app-prob-der}
    d \cdot p_i^\downarrow = \eval{f'_{\v{p}}(x)}_{\frac{i-1}{d}< x < \frac{i}{d}} \equiv f'_{\v{p}}(x_i),
\end{equation}
where we define $x_i\in \qty[(i-1)/d,i/d]$ as an arbitrary point from the interval between the consecutive elbows of the majorisation curve. Using the relation given by Eq.~\eqref{Eq:app-prob-der} in Eqs.~\eqref{eq:a_inft_ver}~-~\eqref{eq:b_inft_ver}, we find that the coefficients $a$ and $b$ can be expressed in terms of derivatives as follows:
    \begin{align}
        a & = \min\qty(\frac{f'_{\v{p}}(x_1)}{f'_{\v{p}}(x_{m})},\frac{f'_{\v{p}}(x_{n+1})}{f'_{\v{p}}(x_{d})}) > \max_{v\in\qty{1,2,\hdots,k-1}}\qty(\frac{f'_{\v{r}}(x_\nu)}{f'_{\v{r}}(x_{\nu+1})}),\label{eq:a_der_ver}\\
        b & = \max_{l\in\mathcal{L}} \frac{f'_{\v{p}}(x_l)}{f'_{\v{p}}(x_{l+1})} < \frac{f'_{\v{r}}(x_1)}{f'_{\v{r}}(x_{k})}.\label{eq:b_der_ver}
    \end{align}
Now, the expressions above remain valid when we replace the subset $\mathcal{L}$, over which we maximize, with an open interval $\mathcal{L} = (m',n') = \qty{x: f'_{\v{p}}(x)-f'_{\v{q}}(x) < 0}$;  the same substitution can be applied to the right-hand side~\eqref{eq:a_der_ver}. However, to maintain the correct interpretation, we must substitute the ratios of subsequent points with the ratios of the left and right derivatives, specifically focusing on the left-hand side of Eq.~\eqref{eq:b_der_ver}. As result, leads to the following 
\begin{equation}
    b = \max_{x\in\mathcal{L}} \frac{f'_{\v{p}}(x_-)}{f'_{\v{p}}(x_+)}.
\end{equation}
where we used shorthand notation for left and right derivatives $f'(x_{\pm})\equiv \lim_{y\rightarrow x_\pm} f'(x)$. Note that, whenever a function is differentiable at point $x$, the ratio between left and right derivative is equal to one. Therefore, any ratio value differing from one can be considered indicative of non-differentiability. This is a characteristic observed exclusively at the elbows of majorisation curves—a property we will later exploit to further simplify our calculations.

Finally, we return to the discussion of thermomajorisation curves by reintroducing $\beta > 0$, which establishes the connection between the derivatives and probabilities.
\begin{equation}
    p_i^\beta = \eval{\gamma_i^\beta f'_{\v{p}}(x)}_{\Gamma_{i-1}< x < \Gamma_i},
\end{equation}
where we used the subsums of the Gibbs distribution $\Gamma_i = \sum_{j=1}^i \gamma_j^\beta$, which correspond to the locations of the elbows on the thermomajorisation curve $f_{\v{p}}(x)$. While all inequalities remain valid, they now depend not only on the probability vectors $\v{p}$ and $\v{r}$ but also on the respective underlying Gibbs distributions $\v{\gamma}$ and $\v{\gamma}_{\v{r}}$.

Now, let us focus on the derivation of lower bound on catalyst dimensionality $k_*$ from the derivatives' perspective. In the original approach the starting point was to note the following chain of bounds,
\begin{equation}
    b < \frac{r^\downarrow_1}{r^\downarrow_k} = \prod_{i=1}^{k-1} \frac{r^\downarrow_i}{r^\downarrow_{i+1}} < \qty[\max_{v\in\qty{1,2,\hdots,k-1}}\qty(\frac{r^\downarrow_v}{r^\downarrow_{v+1}})]^{k-1}<a^{k-1}.
\end{equation}
Unfortunately, the heuristic approach does not allow us to extend this result directly, as we shifted our perspective from piecewise-linear curves defined solely by their elbows to continuous functions. However, by taking the logarithm of Eq.~\eqref{eq:b_der_ver}, one can observe that it involves the ratio between the first and last slope. This allows for the computation of the left-derivative at $x_1=0$ and the right-derivative at $x_k=1$,

\begin{equation}
    \log(\frac{f'_{\v{r}}(0_+)}{f'_{\v{r}}(1_-)}) = \log[f'_{\v{r}}(0_+)] - 
    \log[f'_{\v{r}}(1_-)].
\end{equation}
The above can also be reformulated as the integral of a derivative,
\begin{align}\label{Eq:app-log}
    \log(\frac{f'_{\v{r}}(0_+)}{f'_{\v{r}}(1_-)}) & = \lim_{\delta\rightarrow0}\int_\delta^{1-\delta} \qty[-\log(f'_{\v{r}}(x))]'\dd{x}. 
\end{align}
Next, we observe that the logarithmic terms appearing in Eq.~\eqref{Eq:app-log} can be explicitly expressed as
\begin{align}
    \log(f'_{\v{r}}(x)) =& \sum_{i=1}^k \log(s_i(\v{r})) \qty(\Theta\qty(x - \Gamma_{i-1})-\Theta\qty(x - \Gamma_{i})) \label{Eq:app-log-1}, \\
    \qty[\log(f'_{\v{r}}(x))]'  =&  \sum_{i=1}^k \log(s_i(\v{r})) \qty(\delta\qty(x - \Gamma_{i-1})-\delta\qty(x - \Gamma_{i})) \nonumber \\
    = & \log(s_1(\v{r}))\delta(x) - \log(s_d(\v{r}))\delta(x-1)\nonumber \\
    + & \sum_{i=1}^{k-1} \delta(x-\Gamma_i)\qty[\log(s_{i+1}(\v{r})) - \log(s_i(\v{r}))]. \label{Eq:app-log-2}
\end{align}
In Eqs.~\eqref{Eq:app-log-1} and \eqref{Eq:app-log-2}, $\Theta(x)$ represents the Heaviside step function, and $\delta(x)$ denotes the Dirac delta function, respectively.

Thus, using Eqs.\eqref{Eq:app-log-1} and \eqref{Eq:app-log-2}, one can compute the integral given in Eq.\eqref{Eq:app-log} as follows:
\begin{align}
    \lim_{\delta\rightarrow0}\int_\delta^{1-\delta} 
    \!\!\!\!\!\!\!\!
    \qty[-\log(f'_{\v{r}}(x))]'\dd{x} & =
    \sum_{i=1}^{k-1} \log(\frac{s_i(\v{r})}{s_{i+1}(\v{r})}) \nonumber \\
    & \leq (k-1)\max_i \log(\frac{s_i(\v{r})}{s_{i+1}(\v{r})}).
\end{align}
where we recall that $\v{s}(\v{r})$ is the slope vector corresponding to the population of the catalyst $\v{r}$. Finally, using~\eqref{eq:a_der_ver}, we retrieve the bound on the dimensionality of the catalyst state $\v{r}$, 
\begin{equation}
    k \geq \frac{\log b}{\log a} + 1 = k_*,
\end{equation}
within the heuristic approach. 

\subsection{Embedding map approach}

We start by recalling the embedding map~\cite{horodecki2013fundamental}~(see Ref.~\cite{Lostaglio2019} for a detailed discussion). This allows us to draw a connection between thermomajorisation and majorisation in a space of larger dimension.

\begin{defn}[Embedding map]\label{def:emb_map}
Consider a thermal distribution $\v{\gamma}$ with rational entries, \mbox{$\gamma_i=D_i/D$} and \mbox{$D_i,D\in\mathbb{N}$}, the embedding map $\Xi$ sends a $d$-dimensional probability distribution $\v{p}$ to a $D$-dimensional probability distribution $\hat{\v{p}}:=\Xi(\v{p})$ as follow:
	\begin{equation}
	\hat{\v{p}}=\Biggl[
    \underbrace{\frac{p_1}{D_1},\hdots,\frac{p_1}{D_1}}_{D_1 \: \text{times}},
    \underbrace{\frac{p_2}{D_2},\hdots,\frac{p_2}{D_2}}_{D_2 \: \text{times}},\hdots,
    \underbrace{\frac{p_d}{D_d},\hdots,\frac{p_d}{D_d}}_{D_d \: \text{times}}\Biggr]\label{eq:embedding}.
	\end{equation}
\end{defn}
\noindent A generic thermal distribution $\boldsymbol{\gamma}$ will not fall within the scope of Definition~\ref{def:emb_map}. Nevertheless, since the rational numbers $\mathbb{Q}$ are dense subset of the real numbers $\mathbb{R}$. Thanks to this the limit $\lim_{D\rightarrow\infty}  \frac{\lceil D \gamma_i\rceil}{D} = \lim_{D\rightarrow\infty}  \frac{\lfloor D \gamma_i\rfloor}{D} = \gamma_i$ exists. With this chain of approximations in mind, we proceed to work within the limit of $D\rightarrow\infty$.

It follows from the above definition that the majorisation between embedded vectors coincides exactly with the notion of thermo-majorisation
\begin{equation}
 \hat{\v{p}} \succ \hat{\v{q}} \Longleftrightarrow  \v{p}\succ_\beta\v{q}
\end{equation}
Now, we introduce a function $\varphi$ that maps elements of the embedded distribution to the corresponding elements of the original one
\begin{equation}
    \hat{p}^\downarrow_{j} = \frac{p_{\varphi(j)}}{D_{\varphi(j)}},
\end{equation}
which, as a consequence of non-increasing ordering, yields
\begin{equation}
    j>k \Rightarrow \frac{p_{\varphi(j)}}{D_{\varphi(j)}} \leq \frac{p_{\varphi(k)}}{D_{\varphi(k)}}.
\end{equation}
The function $\phi$ is used in order to convert the ordering of the embedded vector $\hat{\v{p}}^\downarrow$ to the $\beta$-ordering of $\v{p}^\beta$. Suppose, for example, that $p^\beta_1 = p_3$. Knowing this, we find that $\phi(j) = 3$ for all $1\leq j\leq D_3$. More generally, it is defined to be, we find that $\phi(j) = \pi^{-1}(i)$ for $\sum_{n=1}^i D_{\pi^{-1}(n)}+1\leq j \leq \sum_{n=1}^i D_{\pi^{-1}(n)}$.

Focusing on the expression \eqref{eq:b_inft_ver} for $b$ and applying it to the embedded vector $\hat{\v{p}}$, we substitute

\begin{align}
    \frac{\hat{p}^\downarrow_l}{\hat{p}^\downarrow_{l+1}} & = \frac{p_{\varphi(l)}}{D_{\varphi(l)}}\frac{D_{\varphi(l+1)}}{p_{\varphi(l+1)}} \\
    & = \frac{p_{\varphi(l)}}{D_{\varphi(l)}/D}\frac{D_{\varphi(l+1)}/D}{p_{\varphi(l+1)}} \\ & \overset{D\rightarrow\infty}{\longrightarrow} \frac{p_{\varphi(l)}}{\gamma_{\varphi(l)}}\frac{\gamma_{\varphi(l+1)}}{p_{\varphi(l+1)}} = \frac{s_{\pi(\varphi(l))}(\v{p})}{s_{\pi(\varphi(l+1))}(\v{p})}
\end{align}
where in the last expression $\pi$ is the permutation responsible for the beta-ordering of $\v{p}^\beta$. Treating all the remaining ratios in \eqref{eq:a_inft_ver}, \eqref{eq:b_inft_ver}, we retrieve the full form of Theorem~\ref{thm:cat_beta}.

\section{Upper-bounding number of intersections between d-point simplices} \label{app:geometric_fancy}

We start by considering a regular $d$-point simplex $\Delta_d$ spanned by vertices $\v{v}_i$ for $i = \{1,\hdots,d\}$, with the central point $\v{c} = \frac{1}{d}\sum_i\v{v}_i$. In what follows, we will consider the rescaled versions of the simplex defined as
\begin{equation}
    t\Delta_d = \operatorname{conv}\qty(\qty{t \v{v}_i + (1-t) \v{c}}_{i=1}^d).
\end{equation}

First, consider the intersection $\Delta_d \cap t\Delta_d$ for $t\leq 0$ which, in fact, defines an inverted or dual simplex. This will be denoted as $t_d \nabla_d$ with $t_d = -t$, a choice whose rationale will soon become apparent. For $t_d\in[0,1/(d-1)]$, we observe that there is no intersection between the boundaries, $\partial\Delta_d\cap \partial (t_d\nabla_d) = \emptyset$. The first intersection occurs at $t = -1/(d-1)$, at which point there is only a single extreme point per $d$-dimensional face of $\Delta_d$, since
\begin{equation}
    \frac{1}{d-1} \nabla_d= \operatorname{conv}\qty(\qty{\sum_{j=1}^d\frac{1-\delta_{j}}{d-1}\v{v}_j}_{i=1}^d).
\end{equation}
For $t_d \geq 1/(d-1)$, the intersection points of $-t\nabla_d$ with each of the $d$ faces of $\Delta_d$—which are effectively $(d-1)$-dimensional simplices—result in a $(d-1)$-dimensional inverted simplex, denoted as $t_{d-1}\nabla_{d-1}$. Here, we define $t_{d-1} = \frac{d-1}{d}\min(0, t_d - \frac{1}{d-1})$. Consequently, there are exactly $(d-1)$ intersections per face of $\Delta_d$, continuing until $t_{d-1} = 1/(d-2)$. The analysis then continues for the $(d-2)$-dimensional faces of $\Delta_d$.

Based on the reasoning above, we can identify only two types of situations:
\begin{enumerate}
    \item The intersection $\Delta_d \cap t\nabla_d$ has a single vertex per $k$-dimensional face of $\Delta_d$.
    \item The intersection $\Delta_d \cap t\nabla_d$ has exactly $(d-k)$ vertices per $k$-dimensional face of $\Delta_d$.
\end{enumerate}
This results in either $\binom{d}{k}$ vertices or $k\binom{d}{k}$ vertices in total, with the latter being the larger one and attaining maximum at $k = \left\lceil\frac{d}{2}\right\rceil$. This completes the reasoning for simplices with a common center $\v{c}$. 

The heuristic approach for the translations defined as
\begin{equation}
    \Delta_d + \v{\delta} = \operatorname{conv}\qty(\qty{\v{v}_i + \v{\delta}}_{i=1}^d).
\end{equation}
is based on the following reasoning. Let us focus on the intersection with a single face $F$ from $\Delta_d$ and assume that $t_{d-1}< \frac{2}{d-2}$. Any shift in a direction parallel to the plane the face is contained in, $\v{\delta}\parallel F$, will result in translating the intersection points between $t\nabla_d$ and $F+\v{\delta}$. Similarly, the shift in the perpendicular direction will result in scaling the ``intersection simplex''. This observation indicates that the general structure of the intersections remains unchanged, confirming that the upper bound on the number of vertices from concentric $\Delta_d$ and $t_d\nabla_d$ holds.

\section{Proof of Proposition \ref{prop:peculiar_point}} \label{app:peculiar_point}

In this section, we assume that the Gibbs state is given by $\v{\gamma}  =\qty(\gamma_1,\,\gamma_2,\,\gamma_2,\,\gamma_4)$. As a result, arbitrary unitary operation in the subspace $\operatorname{span}(\ket{01},\,\ket{10})$ is energy-preserving, and thus an admissible thermal operation (TO). Under these assumptions there exists a subset of states that can become entangled under TOs. In particular, for 2-qubit states with populations $\v{p}$ a state can be entangled using energy-preserving unitaries if and only if $4p_1p_4 - (p_2 - p_3)^2 < 0$. Furthermore, as shown in \cite{de2024ent}, a state $\v{p}$ can be entangled by TOs if and only if the extreme point of its future thermal cone $\v{p}^\pi \in \mathcal{T}_+(\v{p})$ with $\beta$-order given by $\pi = (2134)$. We denote by $\mathbb{TN}$ the set of states $\v{p}$ for which no point of the future cone $\mathcal{T}_+(\v{p})$ can become entangled under energy-preserving unitaries; this set is called thermally non-entanglable set. Similarly, we define a set of catalytically nonentaglable set $\mathbb{CN}$ as set of states for which no state from either their future cone or the catalysable set $\mathcal{C}_+(\v{p})$ can become entangled under energy-preserving unitaries.

In order to prove Proposition~\ref{prop:peculiar_point}, we need to show that the boundary of set $\mathbb{CN}$ always contains a subspace-thermalised state of the form
\begin{equation}
    \partial \mathbb{CN}\ni\v{p}^* = (1-t')\qty(0,0,0,1) + \frac{t'}{Z_{123}} \qty(\gamma_1,\gamma_2,\gamma_2,0),
\end{equation}
with $Z_{123} = \gamma_1 + 2\gamma_2$. This state by itself is clearly non-entanglable by using only energy-preserving unitaries, as it is proportional to identity on the subspace $\operatorname{span}(\ket{01},\,\ket{10})$ and as such, does not admit generation of any coherences within the aforementioned subspace. However, its future thermal cone $\mathcal{T}_+(\v{p}^*)$ may contain states that are entanglable.

In particular, one of the extreme points of $\mathcal{T}_+(\v{p}^*)$ is a subspace-thermalised state of the form
\begin{equation}
    \v{p}^{**} = (1-t)\qty(0,1,0,0) + \frac{t}{Z_{124}} \qty(\gamma_1,0,\gamma_2,\gamma_4),
\end{equation}
with $Z_{123} = \gamma_1 + \gamma_2 + \gamma_4$. This admits the distinguished $\beta$-order of $\v{\pi} = (2,1,3,4)$, and thus, by Theorem 2 from \cite{de2024ent}, if this is not an element of the entanglable set, $\v{p}^{**}\in\mathbb{TN}$ It is easy to find $t$ for which this state becomes non-entanglable by solving,
\begin{equation}
    4\frac{t^2}{Z_{124}^2}\gamma_1\gamma_4 - \qty(1-t-\frac{t}{Z_{124}}\gamma_2)^2 = 0,
\end{equation}
which is solved by setting 
\begin{equation}
    t = \frac{\gamma_1+\gamma_2+\gamma_4}{2 \sqrt{\gamma_1}
   \sqrt{\gamma_4}+1}.
\end{equation}

Now we would like to find a state $\v{p}^*$
which tightly majorises the state $\v{p}^{**}$, $\v{p}^*\succcurlyeq\v{p}^{**}$. It is achieved by setting $t' = \frac{Z_{123}}{Z_{124}}$. By reinstating $\gamma_1 = 1,\,\gamma_2 = \gamma_3 = e^{-\beta},\,\gamma_4 = e^{-2\beta}$, we retrieve Proposition \ref{prop:peculiar_point}.

According to the above derivation, we can say that $\mathcal{T}_+(\v{p}^*)\subset\mathbb{TN}$, but also $\mathcal{T}_+(\v{p}^*)\subset\mathbb{CN}$, as $\mathcal{C}_+(\v{p}^*) = \emptyset$ and for all $\v{q}\prec\v{p}^*$ we find that $\mathcal{C}_+(\v{q})\subset\mathcal{T}_+(\v{p}^*)$. However, by explicit check one can verify that for $\v{q}'\in\mathcal{B}(\v{p}^*,\epsilon)\cap\mathcal{T}_+(\v{p}^*)$ and $\epsilon\leq \epsilon_*$, where $\mathcal{B}(\v{p}^*,\epsilon)$ is a ball of radius $\epsilon$ centered at $\v{p}^*$ and $\epsilon_*\ll1$ a certain critical value, we find $\v{q}' \in \mathbb{CE}$.

\bibliography{biblio}
\end{document}